\documentclass[11pt]{article}
\usepackage{thmtools, thm-restate}
\usepackage{macros}

\title{Robust Gray Codes Approaching the Optimal Rate} 
\author{
Roni Con\thanks{Department of Computer Science, Technion - Israel Institute of Technology,
                    Haifa, Israel,
                    roni.con93@gmail.com.}
\and 
Dorsa Fathollahi\thanks{Department of Electrical Engineering,
                    Stanford University,
                    Stanford, CA,
                    dorsafth@stanford.edu.}
\and
Ryan Gabrys\thanks{University of California San Diego, San Diego, CA, rgabrys@ucsd.edu. }
\and
Mary Wootters\thanks{Department of Electrical Engineering,
                    Stanford University,
                    Stanford, CA,
                    marykw@stanford.edu.}
\and
Eitan Yaakobi\thanks{Department of Computer Science, Technion - Israel Institute of Technology,
                    Haifa, Israel,
                    yaakobi@cs.technion.ac.il}
}
\date{\today}

\begin{document}

\maketitle

\begin{abstract}
Robust Gray codes were introduced by (Lolck and Pagh, SODA 2024).  Informally, a robust Gray code is a (binary) Gray code $\mathcal{G}$ so that, given a \emph{noisy} version of the encoding $\mathcal{G}(j)$ of an integer $j$, one can recover $\hat{j}$ that is close to $j$ (with high probability over the noise).  Such codes have found applications in differential privacy.

In this work, we present near-optimal constructions of robust Gray codes.  
In more detail, we construct a Gray code $\mathcal{G}$ of rate $1 - H_2(p) - \varepsilon$ that is efficiently encodable, and that is robust in the following sense.  Supposed that $\mathcal{G}(j)$ is passed through the binary symmetric channel $\text{BSC}_p$ with cross-over probability $p$, to obtain $x$.  We present an efficient decoding algorithm that, given $x$, returns an estimate $\hat{j}$ so that $|j - \hat{j}|$ is small with high probability.\blfootnote{

DF is partially supported by NSF grant CCF-2133154.
The work of RG was partially supported by  NSF Grant CCF-2212437. MW is partially supported by NSF grants CCF-2133154 and CCF-2231157.
The work of RC and EY was supported by the European Union (DiDAX, 101115134). 
Views and opinions expressed are those of the author(s) only and do not necessarily reflect those of the European Union or the European Research Council Executive Agency. Neither the European Union nor the granting authority can be held responsible for them.\\
Part of this work was done while the authors were visiting the Simons Institute for the Theory of Computing. 
}

\end{abstract}

\newpage

\section{Introduction}

A \emph{robust Gray code} is a Gray code that is robust to noise.  In more detail, a robust Gray code $\calg$ of length $d$ is a map $\calg:\{0, \ldots, N-1\} \to \{0,1\}^d$ so that: 
\begin{itemize}
    \item $\calg$ is a Gray code: For all $j \in \{0, \ldots, N-1\}$, $\Delta(\calg(j), \calg(j+1)) = 1$, where $\Delta$ denotes the Hamming distance.\footnote{The paper \cite{LP24} also gives a more general definition, where the code should have low \emph{sensitivity}, meaning that $|\enc{\calg}(j) - \enc{\calg}(j+1)|$ is small; however, both their code and our code is a Gray code, so we specialize to that case (in which the sensitivity is $1$).}
    \item $\calg$ is robust to noise from the \emph{binary symmetric channel with cross-over probability $p$} ($\BSCp$), for some $p \in (0,1/2)$: Let $\eta \sim \ber(p)^d$ be a random noise vector.  Then for any $j \in \{0, \ldots, N-1\}$, given $\calg(j) \oplus \eta$, it should be possible to (efficiently) recover an estimate $\hat{j}$ so that $|j - \hat{j}|$ is small, with high probability over $\eta$. 
\end{itemize}
As with standard error-correcting codes, we define the \emph{rate} of a robust Gray code $\calg:\{0,\ldots, N-1\} \to \{0,1\}^d$ by $\calr = \frac{\log_2(N)}{d}$.  The goal is then to make the rate as high as possible while achieving the above desiderata.

For intuition about the problem, consider two extreme examples.  The first is the \emph{unary code} of length $d=N$.  The unary code simply encodes an integer $j$ as $j$ ones followed by $d - j$ zeros.  It is not hard to see that if some random noise is introduced (with $p < 1/2$), it is possible to approximately identify $j$; it is the place where the bits go from being ``mostly one'' to ``mostly zero.'' However, the rate of this code tends to zero very quickly; it has rate $\calr = \log_2(d)/d$. The second extreme example is the classical Binary Reflected Code (BRC, see \Cref{def:BRC}).  The BRC is a Gray code with $N = 2^d$ and hence rate $\calr = 1$, as high as possible.  However, the BRC is not at all robust. For example, the encodings of $0$ and $N-1$ under the BRC differ by only a single bit, and more generally changing the ``most significant bit'' (or any highly significant bit) can change the value encoded by quite a lot.  Our goal is something in between these extreme examples: A Gray code with good rate (as close to $1$ as possible), but also with good robustness.  Geometrically, one can think of this as a path that ``fills up'' as much of the Boolean cube $\zo^d$ as possible, while not getting too close to distant parts of itself too often.

Robust Gray codes were introduced by Lolck and Pagh in~\cite{LP24}, motivated by applications to differential privacy.  While their particular application (to differentially private histograms) is a bit involved, the basic idea is the following.  In differential privacy, one adds noise to protect privacy, while hoping to still be able to estimate useful quantities about the data.  Adding continuous noise (say, Laplace noise) to real values is standard, but it can be more practical to add noise from the $\BSCp$ to binary vectors.  This motivates a robust Gray code as a building block for differentially private mechanisms: It is a way of encoding integer-valued data into binary vectors, so that the original value can be estimated after noise from the $\BSCp$ is added.

The original paper of Lolck and Pagh introduced a construction of robust Gray codes that transformed any binary error-correcting code $\calc$ with rate $\calr$ into a robust Gray code $\calg$ with rate $\Omega(\calr)$. They showed that if $\calc$ had good performance on $\BSCp$, then so did $\calg$; more precisely, given $\calg(j) \oplus \eta$, their decoder produced an estimate $\hat{j}$ so that
\[ \Pr_\eta[ |j - \hat{j}| \geq t ] \leq \exp( -\Omega(t)) + \exp(-\Omega(d)) + O(\pfail(\calc)), \]
where $\pfail(\calc)$ is the failure probability of $\calc$ on the $\BSCp$.  
However, the constant in the $\Omega(R)$ in the rate in that work is at most $1/4$, which means that it is impossible for the construction of \cite{LP24} to give a high-rate code, even if $p$ is very small.  The constant inside the term $\Omega(\calr)$ was improved to approach $1/2$ in \cite{fathollahi2024improved}.\footnote{The work \cite{fathollahi2024improved} is by a subset of the authors of the current paper; we view it as a preliminary version of this work.}

Our main result is a family of robust Gray codes that have rate approaching $1 - H_2(p)$ on the $\BSCp$, where $H_2(p) = -p\log_2(p) - (1-p)\log_2(1-p)$ is the binary entropy function.  In fact, we prove a more general result, which takes any binary linear code $\Cin$ of rate $\calr$, and transforms it into a robust Gray code $\calg$ with rate approaching $\calr$.  This more general result is stated in Theorem~\ref{thm:main} below; we instantiate it in Corollary~\ref{cor:ach-capacity} to achieve rate approaching $1 - H_2(p)$.

\begin{restatable}{theorem}{mainResult}
\label{thm:main}
    Fix constants $p \in (0,1/2)$ and a sufficiently small $\varepsilon > 0$.  Fix a constant $\calr \in (0,1)$.
    Let $d$ be sufficiently large, in terms of these constants.  Then there is an $n' = \Theta(\log d)$ so that the following holds. Suppose that there exists a binary linear $[n',k']_2$ code $\Cin$ with rate $k'/n' = \calr$ so that $\Cin$ has a decoding algorithm $\dec{\Cin}$ that has block failure probability on the $\BSCp$ that tends to zero as $n' \to \infty$.\footnote{See \Cref{def:prob-fail} for a formal definition of the failure probability.}  
    Then there is a robust Gray code $\calg: \{0,1,\ldots N-1\} \to \F_2^d$ and a decoding algorithm $\dec{\calg}:\F_2^d \to \{0,1,\ldots N-1\}$ so that:
    \begin{enumerate}
        \item The rate of $\calg$ is $\calr - \varepsilon$. 
        \item Fix $j\in \{0,1,\ldots N-1\}$, let $\eta \sim \ber(p)^d$ be a random error vector, and let $\jhat:= \dec{\calg}(\calg(j) \oplus \eta)$, where $\eta \sim \ber(p)^d$. Then 
        \[
            \Pr_\eta[|j-\jhat| \geq t ] \leq \exp(-\Omega(t)) + \exp\left(-\Omega\left(\frac{d}{\log d}\right)\right) \;,
        \]
        where the constants inside the $\Omega(\cdot)$ notation depend on $p, \varepsilon$, and $\calr$.
        \item The running time of $\calg$ (the encoding algorithm) is $\tilde{O}(d^3)$ and the running time of $\dec{\calg}$ (the decoding algorithm) is $\tilde{O}(d^2)$ where the $\tilde{O}(\cdot)$ notation hides logarithmic factors. 
    \end{enumerate}
\end{restatable}

\begin{remark}[The running time of $\dec{\Cin}$]\label{rem:runningtime}
We note that the running time of $\dec{\Cin}$ does not appear in \Cref{thm:main}.  The reason is that for any code, the brute-force maximum-likelihood decoder runs in time $\mathrm{poly}(n')\cdot 2^{k'}$.  In the proof of \Cref{thm:main}, we will choose $k' = \log(n+1) \leq \log d$, which implies that $n' = O(\log d)$.  Thus, the running time of $\dec{\Cin}$ is at most $d \cdot \mathrm{polylog(d)}$, and this is sufficiently small to obtain the bound on the running time of $\dec{\calg}$ in \Cref{thm:main}.
\end{remark}

For the best quantitative results, we instantiate \Cref{thm:main} by choosing $\Cin$ to be a binary code that achieves capacity on the binary symmetric channel, for example, polar codes~\cite{A08,tal2013construct,GX14,hassani2014finite,guruswami2020arikan,blasiok2022general}; 
Reed-Muller codes~\cite{AS23,RP23}); or even a random linear code.  This yields the following corollary.

\begin{corollary}\label{cor:ach-capacity}
Let $p \in (0,1/2)$, $\epsilon > 0$ be sufficiently small, and fix positive integers $N$ and $d$ sufficiently large, and with $\calr := \frac{\log_2(N)}{d} = 1 - H_2(p) - \epsilon$.  Then there is an efficiently encodable robust Gray code $\cG:[N] \to \F_2^d$ of rate $\calr$, so that the following holds.  There is a polynomial-time algorithm $\dec{\calg}: \F_2^d \to [N]$ so that for any $j \in [N]$, for $\eta \sim \ber(p)^d$, $\hat{j} = \dec{\calg}(\cG(j) \oplus \eta)$ satisfies
\[ 
    \Pr_\eta[|j-\jhat| \geq t ] \leq \exp(-\Omega(t)) + \exp\left(-\Omega\left(\frac{d}{\log d}\right)\right) \;.
\]
for any $t \geq 0$.
\end{corollary}

We note that $1 - H_2(p)$ is the Shannon capacity for $\BSCp$, which implies that the limiting rate of $1 - H_2(p)$ in \Cref{cor:ach-capacity} is  optimal in the following sense.

\begin{observation}[Optimality of \Cref{cor:ach-capacity}]
    Suppose that $\calg:[N] \to \{0,1\}^d$ is a robust Gray code with rate $\calr = 1 - H_2(p) + \theta$ for some constant $\theta > 0$.  Let $\eta \sim \ber(p)^d$ for some $p \in (0,1/2)$.  Then, for any procedure that recovers $\hat{j}$ from $\calg(j) \oplus \eta$ and any $t = \Theta(1)$ and for sufficiently large $N$, we have
    $\Pr_\eta[ |j - \hat{j}| > t ] \geq 0.99.$
\end{observation}
\begin{proof}
    Suppose that $\calg$ is in the statement of the observation, but that $\Pr_\eta[|j - \hat{j}| > t] < 0.99.$
    Then one could use $\calg$ to communicate with non-negligible failure probability on the $\BSCp$ as follows.  The sender will encode a message $j \in \{0,\ldots, N-1\}$ as $\calg(j)$ and send it over the channel.  The receiver sees $\calg(j) \oplus \eta$ and uses $\calg$'s decoding algorithm (possibly inefficiently) to recover $\hat{j}$.  Then the receiver returns $\tilde{j}$ chosen uniformly at random from the interval $I = \{\hat{j}-t, \hat{j}-t + 1, \ldots, \hat{j} + t\}.$
    The success probability of this procedure will be at least $0.01 \cdot \frac{1}{2t+1}$.  Indeed, with probability at least $0.01$, we have that $|j - \jhat| \leq t$ and hence $j \in I$, and, if that occurs, then with probability at least $1/(2t+1)$ we will have $\tilde{j} = j$, as $|I| = 2t + 1$.
    However, the converse to Shannon's channel coding theorem implies that the success probability for any code with rate $\calr$ can be at most $\exp(-\Omega_{\theta,p}(d))$ (see, e.g., \cite[Theorem 1.5]{madhu_notes}).  
    This is a contradiction for sufficiently large $d$ when $t = \Theta(1)$ is a constant (or even polynomial in $d$). 
\end{proof}

\subsection{Related Work}\label{sec:rel}
As mentioned earlier, robust Gray codes were originally motivated by applications in differential privacy, and have been used in that context; see~\cite{LP24,ALP21,ALS22,ACLST21} for more details on the connection. 
Beyond the initial construction of \cite{LP24}, the only prior work we are aware of is that of \cite{fathollahi2024improved} mentioned above, which we build on in this paper.
It is worth mentioning that there exist non-binary codes based on the Chinese Remainder Theorem~\cite{XXW20,WX10} that have nontrivial sensitivity, but in our work, we focus on binary codes.

\paragraph{Independent Work.} While this paper was in preparation, it came to our attention that Guruswami and Wang have achieved similar results, but with different techniques~\cite{GW24}. In particular, their approach does not use code concatenation.

\subsection{Technical Overview}\label{sec:tech} 

Before diving into the details, we give an overview of our construction along with a discussion of how our approach leverages (and also departs from) ideas presented in previous work. In \cite{LP24}, the main idea was to transform a linear binary ``Base'' code $\cC_{B}$ with rate $R$ into a robust gray code $\cC_{\cG}$ with rate $\Omega(R)$. The technique used involves first concatenating four copies of a codeword from $\cC_B$, of which two are bit-wise negated, in addition to some padding bits to form a codeword in an intermediate code, denoted $\calw$, that is eventually transformed into the code $\cC_{\cG}$. Since each codeword in $\calw$ (and also $\cC_{\cG}$) is composed of four copies of $x \in \cC_B$, it is possible, even in the presence of noise, to allow one of the copies of $x$ to be unrecoverable and still be able to use majority logic on the other three copies to determine the value of the encoded information. 

In our preliminary version of this paper \cite{fathollahi2024improved}, we were able to use an ordering of $\calw$, itself based on a Gray code, that allows us to construct each codeword in $\calw$ using only \textit{two} copies of a given codeword from $\cC_B$, establishing that the rate $R/2$ is achievable. Under this setup, the $i$th codeword in $\calw$ had the following format:
\begin{align*}
b_i \circ c_i \circ b_i \circ c_i \circ b_i,
\end{align*}
where $c_i \in \cC_B$ and where the  $b_i$ is a short padding sequence. However, it remained an open problem to determine whether it is possible to develop a general technique that converts a base code $\cC_B$ of rate $R$ to a robust gray code whose rate also approaches $R$. In this work, we provide an affirmative answer to the previous question. In order to develop such a technique, we rely on two simple ideas. The first idea is to define our base code $\cC_B$ to be a concatenated coding scheme whose resulting code has certain performance guarantees on the $\BSCp$. The second idea has to do with the use of the padding bits. Rather than place our padding bits $b_i$ in between different copies of $c_i \in \cC_B$ to constuct codewords from $\calw$, we will instead embed the markers $b_i$ at regularly spaced intervals within $c_i$. Both these ideas will be discussed in more details in the following exposition. The full technical details of the construction are included in Section~\ref{sec:defs}.

Before we get into a more detailed overview, we define the ingredients we will need.  We require two codes $\Cout$ and $\Cin$ that are compatible under a concatenated error-correcting code scheme, meaning that the parameters are such that the concatenated code $\cC = \Cout \circ \Cin$ makes sense.  We will choose the outer code $\Cout \subseteq \F_q^n$ to be high-rate linear $[n,k]_q$ code, which can correct a small fraction of worst-case errors; and as in \Cref{cor:ach-capacity}, we will choose the inner code $\Cin$ to be any binary code that achieves capacity on the $\BSCp$.

\paragraph{``Interpolating'' between codewords of an intermediate code.} We follow the same high-level idea as in~\cite{LP24,fathollahi2024improved}, in that we first construct an \emph{intermediate code} $\calw$.  The code $\calw$ is a binary code constructed from $\Cout$ and $\Cin$, along with some bookkeeping information; we will describe it in the next paragraph.  We will define an ordering $w_0, w_1, w_2, \ldots$ on the codewords of $\calw$.  Then we will create our final code $\calg$ by ``interpolating'' between the codewords of $\calw$, in order.  We begin by defining $\calg(0) = w_0$. 
 Now, suppose that $z \in [d]$ is the first location that $w_0$ and $w_1$ differ; we define $\calg(1)$ by flipping that bit in $w_0$.  We continue in this way, flipping bits to interpolate between $w_0$ and $w_1$, and then between $w_1$ and $w_2$, and so on. We will choose parameters so that this will {generate distinct encodings for each of our $N$ codewords in the resulting Gray code.}

 \paragraph{Defining the intermediate code.}  While the high-level approach is similar to that in \cite{LP24}, as discussed in the beginning of this section, the improvements come from the definition of the intermediate code $\calw$.  We define it formally in Definition~\ref{def:calw}, but here we give some intuition for the construction.  We begin with an ordering on the codewords $c_0, c_1, \ldots, c_{|\Cout|}$ of the concatenated code $\cC \subseteq \F_2^{n'n}$.  This ordering (formally defined in Section~\ref{sec:base}) has the property that to get from the codeword $c_{i-1}$ to the codeword $c_{i}$, one must simply add one row of the generator matrix $A$ of $\cC$.

 Now, to construct the $i$th codeword $w_i$ in $\calw$, we proceed as follows.  Let $b_i \in \{0,1\}$ be $0$ if $i$ is even and $1$ if $i$ is odd, and let $\vec{b_i}$ denote the bit $b_i$ repeated many times.\footnote{The number of times it is repeated is $B$, the distance of the inner code.  Since the inner code has short length, $\vec{b_i}$ is also not very long, relative to $n$.}  
 Because of our ordering on $\cC$, the only information we need to describe how to transition from $c_{i-1}$ to $c_i$ is the index of which row of $A$ we must add; call this index $z_i \in [kk']$.
 Let $L_{z_i}$ denote an encoding under the repetition code of this information $z_i$; since $z_i$ is short, $L_{z_i}$ can still be fairly short and also be extremely robust against the $\BSCp$.
 Consider a codeword $c_i \in \Cout \circ \Cin$.  This codeword begins with a codeword $\sigma_i \in \Cout$, and has the form
 \[ c_i = c_i[1] \circ c_i[2] \circ \cdots \circ c_i[n],\]
 where $\circ$ denotes concatenation and where for all $m \in \{1, \ldots, n\}$, we have
 \[ c_i[m] = \Cin(\sigma_i[m]) \in \Cin. \]
 We will arrange these inner codewords $c_i[m] \in \Cin$ along with the quantities $L_{z_i}$ and $b_i$ in the following way:

 \begin{center}
     \begin{tikzpicture}[xscale=1.1]
         \draw (0,0) rectangle (12.5,1);
         \node at (-.5,.5) {$w_i = $};
         \draw[fill=red!10] (0,0) rectangle (1.5,1);
         \node at (.75,.5) {$L_{z_i}$};
         \draw[fill=blue!20] (1.5, 0) rectangle (2,1);
         \node at (1.75, .5) {$\vec{b_i}$};
         \draw[fill=yellow!20] (2,0) rectangle (3,1);
         \node at (2.5, .5) {$c_i[1]$};
         \draw[fill=blue!20] (3, 0) rectangle (3.5,1);
         \node at (3.25, .5) {$\vec{b_i}$};
         \draw[fill=yellow!20] (3.5,0) rectangle (4.5,1);
         \node at (4, .5) {$c_i[2]$};
         \draw[fill=blue!20] (4.5, 0) rectangle (5,1);
         \node at (4.75, .5) {$\vec{b_i}$};
         \draw[fill=yellow!20] (5,0) rectangle (6,1);
         \node at (5.5, .5) {$c_i[3]$};
          \draw[fill=blue!20] (12, 0) rectangle (12.5,1);
         \node at (12.25, .5) {$\vec{b_i}$};
         \begin{scope}[xshift=-.5cm]
         \draw[fill=blue!20] (11, 0) rectangle (11.5,1);
         \node at (11.25, .5) {$\vec{b_i}$};
         \draw[fill=yellow!20] (11.5,0) rectangle (12.5,1);
         \node at (12, .5) {$c_i[n]$};
         \draw[fill=blue!20] (9.5, 0) rectangle (10,1);
         \node at (9.75, .5) {$\vec{b_i}$};
         \draw[fill=yellow!20] (10,0) rectangle (11,1);
         \node at (10.5, .5) {\footnotesize $c_i[n\hspace{-.1cm}-\hspace{-.1cm}1]$};
         \draw[fill=yellow!20] (8.5,0) rectangle (9.5,1);
         \node at (9, .5) {\footnotesize $c_i[n\hspace{-.1cm}-\hspace{-.1cm}2]$};
         \node at (7.5,.5) {$\cdots$};
         \end{scope}
     \end{tikzpicture}
 \end{center}

That is, we alternate the inner codewords $c_i[m]$ with bursts of the bit $b_i$, and then include $L_{z_i}$ at the beginning. 

\paragraph{Decoding the resulting robust Gray code.} To see why we define the intermediate code like we do, let us consider what a codeword $\cG(j)$ of our robust Gray code looks like.  Suppose that $\cG(j)$ was an interpolation between $w_i$ and $w_{i+1}$.  Thus, for some ``crossover point'' $h \in [d]$, $\cG(j)$ might look like this:

 \begin{center}
     \begin{tikzpicture}[xscale=1.1]
         \draw (0,0) rectangle (12.5,1);
         \node at (-.7,.5) {$\calg(j) = $};
         \draw[pattern=north west lines, pattern color=orange!60] (0,0) rectangle (1.5,1);
         \node at (.75,.5) {$L_{z_{i+1}}$};
         \draw[pattern=north west lines, pattern color=cyan!20] (1.5, 0) rectangle (2,1);
         \node(c) at (2.5, 2)  {$\vec{b}_{i+1}$};
         \draw[->] (c) to (1.75, .5);
         \draw[->] (c) to (3.25, .5);
         \draw[pattern=north west lines, pattern color=green!40] (2,0) rectangle (3,1);
         \node at (2.5, .5) {\footnotesize $c_{i+1}[1]$};
         \draw[pattern=north west lines, pattern color=cyan!20] (3, 0) rectangle (3.5,1);
         \draw[pattern=north west lines, pattern color=green!40] 
         (3.5,0) rectangle (4,1);
         \draw[fill=yellow!20] (4,0) rectangle (4.5, 1);
         \node(a) at (3.2, -.7) {$c_{i+1}[2]$};
         \draw[->] (a) to (3.75,.5);
         \node(b) at (4.7, -.7) {$c_{i}[2]$};
         \draw[->] (b) to (4.25,.5);
         \draw[fill=blue!20] (4.5, 0) rectangle (5,1);
         \node at (4.75, .5) {$\vec{b_i}$};
         \draw[fill=yellow!20] (5,0) rectangle (6,1);
         \node at (5.5, .5) {$c_i[3]$};
           \draw[fill=blue!20] (12, 0) rectangle (12.5,1);
         \node at (12.25, .5) {$\vec{b_i}$};
         \begin{scope}[xshift=-.5cm]
         \draw[fill=blue!20] (11, 0) rectangle (11.5,1);
         \node at (11.25, .5) {$\vec{b_i}$};
         \draw[fill=yellow!20] (11.5,0) rectangle (12.5,1);
         \node at (12, .5) {$c_i[n]$};
         \draw[fill=blue!20] (9.5, 0) rectangle (10,1);
         \node at (9.75, .5) {$\vec{b_i}$};
         \draw[fill=yellow!20] (10,0) rectangle (11,1);
         \node at (10.5, .5) {\footnotesize $c_i[n\hspace{-.1cm}-\hspace{-.1cm}1]$};
         \draw[fill=yellow!20] (8.5,0) rectangle (9.5,1);
         \node at (9, .5) {\footnotesize $c_i[n\hspace{-.1cm}-\hspace{-.1cm}2]$};
         \node at (7.5,.5) {$\cdots$};
         \end{scope}
        \node[red](h) at (4, 1.5) {$h$};
        \draw[red, ultra thick] (h) to (4, -.5);
     \end{tikzpicture}
 \end{center}
 That is, everything before the ``crossover point'' $h$ has been changed from $w_i$ to $w_{i+1}$.  This picture gives us some intuition for how we should decode $\calg(j) \oplus \eta$ in order to obtain an estimate for $j$.  
 
 The high-level steps in this case would be: 
 \begin{itemize}
     \item \textbf{Identify the approximate location of $h$.}  Observe that the bit $b_i$ is the opposite of the bit $b_{i+1}$.  Thus, with high probability, we can look at the chunks of $\calg(j) \oplus \eta$ that contain the $b$'s and choose a point where they appear to ``switch over'' as an approximation of $h$. 
     \item \textbf{Decode $\Cin$.}  Next, on each chunk that is either $c_{i+1}[r]$ or $c_i[r]$, we run the decoder for $\Cin$ to correctly decode most of them.  This gives us correct estimates for most of the $\sigma_{i+1}[r]$ or $\sigma_i[r]$.
     \item \textbf{Recover a noisy version of $\sigma_i \in \Cout$.} Recall that $L_{i+1}$ contains all the information necessary to recover $c_i$ from $c_{i+1}$ and vice versa.  Thus, after decoding $L_{i+1}$, we can convert all of the $\sigma_{i+1}[r]$'s (at least, those which we have correctly recovered and which we have correctly identified as belonging to $w_{i+1}$ using our estimate of $h$) into $\sigma_i[r]$ for all $r \in [n]$.
     \item \textbf{Decode $\Cout$ to obtain $i$.}  Given our noisy estimates of $\sigma_i[r]$ for all $r \in [n]$, we can now run the decoding algorithm of $\Cout$.  Recall that $\Cout$ can handle a small fraction of worst-case errors; we will show that indeed our estimates of $\sigma_i[r]$ are incorrect for only a small fraction of $r$'s.  After correctly decoding, we can recover $i$.\footnote{In order to recover $i$ efficiently, we leverage the particular ordering that we used on the codewords of $\Cout$.}
     \item \textbf{Recover $\hat{j}$.}  Having correctly identified $i$ and approximately identified $h$ (with high probability), we can now estimate $j$, which is a function only of $i$ and $h$.   
 \end{itemize}
 Of course, there are many more details to be accounted for.  First, one must of course work out the probability of success of all of the above steps, and work out the parameters.  Second, there are several corner cases not captured in the picture above, depending on where the crossover point $h$ lands.  In the rest of the paper, we tackle these details.  In Section~\ref{sec:defs}, we formally define our construction; in Section~\ref{sec:decalg} we state our decoding algorithm; and in Section~\ref{sec:analysis} we analyze it and prove Theorem~\ref{thm:main}.
 
\section{Definitions and Construction}\label{sec:defs}

\subsection{Notation and useful definitions}
We begin with some notation.  For two vectors $x,y$, we use $\Delta(x,y)$ to denote the Hamming distance between $x$ and $y$, and we use $\|x\|$ to denote the Hamming weight of $x$ (that is, the number of non-zero coordinates). 
For an integer $n$, we use $[n]$ to denote the set $\{1, \ldots, n\}$. 
For two strings or vectors, $u$, and $v$ we denote by $u\circ v$ their concatenation. Throughout this paper, we shall move freely between representation of vectors as strings and vice versa. For a string $u$, we define $\text{pref}_m(u)$ to be the prefix of $u$ of length $m$ and similarly $\text{suff}_m(u)$ will denote the last $m$ symbols of $u$. 
For a vector $v$ and an integer $i\geq 1$, we typically use $v[i]$ to denote the $i$th entry of $v$; one exception, defined formally below, is that for a codeword $c$ in the concatenated code $\Cout \circ \Cin$ and for $m \in [n]$, $c[m] \in \Cin$ refers to the $m$th inner codeword in $c$. 

We will use the following versions of the Chernoff/Heoffding bounds.
\begin{lemma}[Multiplicative Chernoff bound; see, e.g., \cite{mitzenmacher2017probability}]\label{lem:chernoff}
    Suppose $X_1, \ldots, X_n$ are independent identically distributed random variables taking values in $\zo$. Let $X = \sum_{i=1}^n X_i$ and $\mu = \bbE[X]$. Then, for any $0<\alpha < 1$:
    \[
    \Pr[X > (1 + \alpha) \mu] < e^{-\frac{\mu  \alpha^2}{3}}
    \]
    and
    \[
    \Pr[X < (1 - \alpha) \mu] < e^{-\frac{\mu  \alpha^2}{2}}
    \]
\end{lemma}

\begin{lemma}[Hoeffding's Inequality; see, e.g., \cite{mitzenmacher2017probability}]\label{lem:hoeffding}
    Suppose $X_1, \ldots, X_n$ are independent random variables (not necessarily identically distributed) taking values in $\pm 1$.  Let $X = \sum_{i=1}^n X_i$ and $\mu = \mathbb{E}[X]$.  Then for any $t \geq 0$, 
    \[ \Pr[ |X - \mu| \geq t ] \leq 2\exp(-t^2/2n).\]
\end{lemma}

Gray codes were introduced in \cite{gray} (see also, e.g.,~\cite{knuth}) which also defined a particular Gray code called the \emph{binary reflected code}. We will use this Gray code to order the codewords in one of our ingredient codes.

\begin{definition}[Binary Reflected Code, \cite{gray}] \label{def:BRC}
        Let $k$ be a positive integer. The \textbf{Binary Reflected Code (BRC)} is a map $\calr_k: \sett{0,\ldots, 2^k-1} \rightarrow \mathbb{F}_2^k $ defined recursively as follows.
    \begin{enumerate}
    \item For $k = 1$,  $\calr_1(0) = 0 $ and $\calr_1(1) = 1 $.
    \item For $k > 1$, for any $i\in \sett{0,\ldots,2^k-1}$, 
    \[ \calr_k(i) = \begin{cases} \calr_{k-1}(i) \circ 0 & i < 2^{k-1} \\ {\calr}_{k-1} ( 2^{k} - i-1) \circ 1 & i \geq 2^{k-1} 
    \end{cases} \]
\end{enumerate}
\end{definition}

Before continuing, we introduce some a few more definitions related to the BRC.  

\begin{definition}\label{def:BRC_stuff}
For $i \in \{1,2,\ldots, 2^{k}-1\}$, let $z_i$ be the unique index where 
\[ \calr_k(i)[z_i] \neq \calr_k(i-1)[z_i].\]
Let $\mathcal{N}_k(z,i)$ be the number of $t \in \{0,1,\ldots, i\}$ so that $z_t = z$.
\end{definition}
That is, the value $z_i$ is the index on which the $i$th codeword in the BRC differs from the previous codeword; equivalently, $z_i$ is the integer for which $\calr_k(i) = \calr_k(i-1) + e_{z_i}$ where $e_{z_i}$ is the $z_i$th unit vector. $\mathcal{N}_k(z,i)$ counts the number of codewords among $\{\calr_k(0), \ldots, \calr_k(i)\}$ that differ from the previous codeword in the $z$th index.  
We give an example of all of this notation below in \Cref{ex:brc}.

\begin{example}\label{ex:brc}
    To illustrate Definitions~\ref{def:BRC} and~\ref{def:BRC_stuff}, we give an example for $k=1,2,3$.
    For $k=1$, we have:
    \begin{center}
        \begin{tabu}{c||cc}
            $i$ & 0&1 \\\hline\hline
          \rowfont{\color{blue}} \textcolor{black}{$\calr_1(i)$} & 0&1
        \end{tabu}
    \end{center}
    For $k=2$, we have:
    \begin{center}
    \begin{tabu}{c||cc|cc}
    $i$ & 0&1&2&3 \\\hline\hline
    \rowfont{\color{blue}} \textcolor{black}{$\calr_2(i)[0]$}&0&1&1&0 \\
       \rowfont{\color{green!50!black}}
    \textcolor{black}{$\calr_2(i)[1]$} &0&0&1&1 \\
    \end{tabu}
    \end{center}
    For $k=3$, we have:
     \begin{center}
    \begin{tabu}{c||cccc|cccc}
    $i$ & 0&1&2&3&4&5&6&7 \\\hline\hline

     \rowfont{\color{blue}} \textcolor{black}{$\calr_3(i)[0]$}&0&1&1&0 & 0&1&1&0 \\
    \rowfont{\color{green!50!black}} \textcolor{black}{$\calr_3(i)[1]$}& 0&0&1&1 & 1&1&0&0 \\
  \rowfont{\color{orange!80!black}} \textcolor{black}{$\calr_3(i)[2]$} &0&0&0&0 & 1&1&1&1\\
    \end{tabu}
    \end{center}
     \setlength{\tabcolsep}{6pt}
     The pattern is that in order to obtain the table for $\calr_k$, we take the table for $\calr_{k-1}$, and repeat it two times, first forwards and then backwards; then we add $\mathbf{0} \circ \mathbf{1}$ as the final row.

Next we give some examples of $z_i$ and $\mathcal{N}_k$.
    For $k=3$, we have the following values of $z_i$:
    \begin{center}
    \begin{tabular}{c||cccccccc}
    $i$ &  1&2&3&4&5&6&7 \\ \hline\hline
    $z_i$ &0&1&0&2&0&1&0 
    \end{tabular}
    \end{center}
    That is, $\calr_3(0) = (0,0,0)$ and $\calr_3(1) = (1,0,0)$ differ in the $z_1 = 0$ component, $\calr_3(1) =(1,0,0)$ and $\calr_3(2) = (1,1,0)$ differ in the $z_2 = 1$ component, and so on. 
    Then, for example, $\mathcal{N}_{k=3}(z=0,i=3) = 2$, as there are two values of $t \leq i$ (names, $t-1$ and $t=3)$ so that $z_t = 2$.  As a few more examples, we have $\mathcal{N}_{k=3}(z=1, i=3) = 1$, and $\mathcal{N}_{k=3}(z=0,i=7) = 4.$
\end{example}

Below in \Cref{obs:BRC}, we state a few useful facts about the $z_i$ and $\mathcal{N}_k(z,i)$.  Briefly, the reason these facts are useful for us is that we will use $\calr_k$ to create the ordering on the codewords $c_i \in \calc$ and $w_i \in \calw$ discussed in the introduction.  Understanding $z_i$ and $\mathcal{N}_k(z,i)$ will be useful for efficiently computing indices in this ordering.

\begin{observation}[Bit Flip Sequence of BRC]\label{obs:BRC} 
For $k \geq 1$, the following holds:
\begin{enumerate}
    \item The index $z_i$ is equal to zero if and only if $i$ is odd.
    \item $\mathcal{N}_k(z,i) = \left\lfloor \frac{i+2^z}{2^{z+1}}\right\rfloor.$
\end{enumerate}
\end{observation}
\begin{proof}
Let $Z_k = (z_1, z_2, \ldots, z_{2^k-1})$, where the $z_t$'s are defined with respect to $k$, as in the statement of the observation.  We first observe that for any $k \geq 2$,
\begin{equation}\label{eq:rec}
Z_k = Z_{k-1} \circ (k-1) \circ Z_{k-1}.
\end{equation}
Indeed, for the base case $k=2$, this follows by inspection: We have $Z_1 = 0$, and $Z_2 = 0,1,0$.  For $k > 2$, it is clear from construction that $Z_k = Z_{k-1} \circ (k-1) \circ \overleftarrow{Z_{k-1}}$, where the $\overleftarrow{\cdot}$ notation means that the sequence is reversed.  However, by induction, $Z_{k-1}$ is symmetric, so we have $\overleftarrow{Z_{k-1}} = Z_{k-1}$.  This establishes the statement for $k$.

Given \eqref{eq:rec}, Item 1 follows immediately by induction.
For Item 2, we proceed by induction on $k$.  As a base case, when $k=1$, the statement follows by inspection.  Now suppose that $k > 2$ and that the statement holds for $k-1$.  

\paragraph{Case 1: $i < 2^{k-1}$.} First suppose that $i < 2^{k-1}$.  Then for any $z < k-1$,
\[ \mathcal{N}_k(z,i) = \mathcal{N}_{k-1}(z,i) = \left\lfloor \frac{i +2^z}{2^{z+1}} \right\rfloor\]
by induction, establishing the statement.  Further, if $z = k-1$ but $i < 2^{k-1}$, we have
\[ \mathcal{N}_k(k-1,i) = 0 = \left\lfloor\frac{i +2^{k-1}}{2^{k}} \right\rfloor,\]
and the statement again follows.
\paragraph{Case 2: $i \geq 2^{k-1}$.} 
Next, we turn our attention to the case that $i \geq 2^{k-1}$.  In this case, by \eqref{eq:rec}, we have
\begin{equation*}
\mathcal{N}_k(z,i) = \mathcal{N}_{k-1}(z, 2^{k-1}-1) + \mathbf{1}[z = (k-1)] + \mathcal{N}_{k-1}(z, i -2^{k-1}).
\end{equation*}
If $z < k-1$, then by induction we have
\[ \mathcal{N}_k(z,i) = \left\lfloor \frac{ 2^{k-1} - 1+ 2^z}{2^{z+1}} \right\rfloor + \left\lfloor \frac{ i - 2^{k-1} + 2^z }{2^{z+1}}\right\rfloor. \]
Suppose that $i = 2^{k-1} + \Delta_1 \cdot 2^{z+1} + \Delta_2$, where $\Delta_2 < 2^{z+1}$.  Then we can write the above as:
\begin{align*}
    \mathcal{N}_k(z,i) &= \left\lfloor 2^{k-z - 2} + \frac{1}{2} - \frac{1}{2^{z+1}}  \right\rfloor + \left\lfloor \Delta_1 + \frac{1}{2} + \frac{\Delta_2}{2^{z+1}}\right\rfloor \\
    &=  2^{k-z - 2} + \Delta_1 + \left\lfloor \frac{1}{2} + \frac{\Delta_2}{2^{z+1}} \right\rfloor,
\end{align*}
where above we have used the fact that $z < k-1$ and so $2^{k-z-2}$ is an integer.  On the other hand, we have
\[ \left\lfloor \frac{i + 2^z}{2^{z+1}} \right\rfloor = 
\left\lfloor 2^{k-z-2} + \Delta_1 + \frac{\Delta_2}{2^{z+1}} + \frac{1}{2} \right\rfloor 
= 2^{k-z-2} + \Delta_1 + \left\lfloor \frac{1}{2} + \frac{\Delta_2}{2^{z+1}}\right\rfloor,\]
which is the same.  Thus, we conclude that if $z < k-1$,
\[ \mathcal{N}_k(z,i) = \left\lfloor \frac{i + 2^z}{2^{z+1}} \right\rfloor,\]
as desired.  On the other hand, if $z = k-1$, then by \eqref{eq:rec} we have $\mathcal{N}_k(k-1,i) = 1$ for all $i \geq 2^{k-1}$, and indeed this is equal to $\left\lfloor \frac{i + 2^{k-1}}{2^k} \right\rfloor$.  This completes the proof of Item 2.
\end{proof}

\begin{definition}[Unary code] The \textbf{Unary code} $\calu \subseteq \F_2^\ell$ is defined as the image of the encoding map $\enc{\calu}:\{0,\ldots, \ell\} \to \F_2^\ell$ given by
$ \enc{\calu}(j) :=  1^{j}\circ 0^{\ell -j }.$
The decoding map $\dec{\calu}: \F_2^\ell \to \{0,\ldots, \ell\}$ is given by
\[ \dec{\calu} (x)  =  \mathrm{argmin}_{j\in {\sett{0,\ldots,\ell}}} \Delta(x ,\enc{\calu}(j)).\]
Similarly, we define the \textbf{complementary Unary code} 
 $\calu^{\text{comp}} \subseteq \F_2^\ell$ as the image of the encoding map $\enc{\calu^{\text{comp}}}:\{0,\ldots, \ell\} \to \F_2^\ell$ given by
$ \enc{\calu^\mathrm{comp}}(j) :=  0^{j}\circ 1^{\ell -j }.$
The decoding map $\dec{\calu^{\mathrm{comp}}}: \F_2^\ell \to \{0,\ldots, \ell\}$ is given by
\[ \dec{\calu^{\text{comp}}} (x)  =  \mathrm{argmin}_{j\in {\sett{0,\ldots,\ell}}} \Delta(x ,\enc{\calu^{\text{comp}}}(j)).\]
\end{definition}
Naively, the runtime complexity of $\dec{\calu}$ is $O(\ell^2)$, as one would loop over $\ell$ values of $j$ and compute $\Delta(x, \enc{\calu}(j))$ for each.  However, this decoder can be implemented in time linear in $\ell$, which is our next lemma.

\begin{lemma} \label{lem:unary-dec-complexity}
    Let $\calu$ be the unary code of length $\ell$. Then $\dec{\calu}$ and $\dec{\calu^{\text{comp}}}$ can be implemented to run in time $O(\ell)$.
\end{lemma}

\begin{proof}
We prove the statement just for $\dec{\calu}$; it is the same for $\dec{\calu^{\mathrm{comp}}}$.
For a fixed $j$, by definition we have $\enc{\calu}(j) = 1^j 0^{\ell-j}$. To compute $\Delta(x, \enc{\calu}(j))$ for each $j$, one needs to count the number of zeros before index $j$ and the number of ones after index $j$. We can express this as follows:
\begin{equation}
\Delta(x, \enc{\calu}(j)) = \sum_{m=1}^{\ell} \mathbbm{1}[x[m] = \enc{\calu}(j)] = \sum_{m=0}^{j} \mathbbm{1}[x[m] = 0] + \sum_{m=j+1}^{\ell} \mathbbm{1}[x[m] = 1]
\end{equation}

Define the array $T[m]$ to count the number of zeros up to index $m$:

\[
    T[m] =\begin{cases}
         \mathbbm{1} [ x [m] = 0] &  m = 1 \\
         T[m-1]+ \mathbbm{1} [x[m] = 0 ] & m>1
    \end{cases}
\]
This array can be computed in time $O(\ell)$. Using $T[m]$, we can rewrite $\Delta(x, \enc{\calu}(j))$ as: 
\begin{equation}
    \Delta (x, \enc{\calu}(j)) = T[j] + (\ell - j - (T[\ell] -  T[j])))
\end{equation}
Thus, given the array $T[m]$, the distance for each $j$ can be computed in $O(1)$ time. Therefore, the overall time complexity of $\dec{\calu}$ is $O(\ell)$.

\end{proof}
Next, we define the failure probability of a binary code.
\begin{definition}\label{def:prob-fail} 
Fix $p \in (0,1)$.
Let $\mathcal{C} \subseteq \mathbb{F}_2^n$ be a code with message length $k$ and encoding and decoding maps $\dec{\mathcal{C}}$ and $\enc{\mathcal{C}}$ respectively.
The probability of failure of $\mathcal{C}$ is
    \begin{equation*}
        \pfail( \mathcal{C}) = \max_{v\in \mathbb{F}_2^k } \Pr_{\eta_p}[ \dec{\mathcal{C}} ( \enc{\mathcal{C}}(v) + \eta_p) \neq v ], 
    \end{equation*}
    where the probability is over a noise vector $\eta_p \in \F_2^n$ with $\eta_p \sim \mathrm{Ber}(p)^n$.
\end{definition}
Note that this definition is simply the block error probability of the binary code $\mathcal{C}$ one the binary symmetric channel with parameter $p$.  
\subsection{Base code}\label{sec:base}

\paragraph{Ingredients.} We begin by fixing an outer code and an inner code. Let $q = 2^{k'}$ for some integer $k'$.
Let $\Cout$ be an $[n,k]_q$ linear code over $\F_q$.

Denote the rate of $\Cout$ by $\Rout \in (0,1)$ and the relative distance of $\Cout$ by $\delout \in (0,1)$. 
Note that it is possible to decode $\Cout$ from $e$ errors and $t$ erasures as long as $2e +t < \delout n$.  Let $\dec{\Cout}:(\F_q \cup \bot)^n \to \F_q^k$ denote the decoding map for $\Cout$ that can do this, where $\bot$ represents an erasure.  (Later, we will choose $\Cout$ to be a Reed--Solomon code, so in particular $\dec{\Cout}$ can be implemented efficiently).

Let $\Cin \subseteq \F_2^{n'}$ be a linear code of dimension $k'$.  We will abuse notation and use $\Cin: \zo^{k'} \rightarrow \zo^{n'}$ to also denote its encoding map.
Let $\Rin = \frac{k'}{n'}$ denote the rate of $\Cin$.
Let $\cC = \Cout \circ \Cin$ denote the concatenation of $\Cout$ and $\Cin$, so that
\[ \cC = \sett{\left( \Cin\left(\sigma[1]\right), \ldots, \Cin\left(\sigma[n]\right) \right)\,:\, \sigma \in \Cout} \subseteq \F_2^{n\cdot n'}, \]
where above we identify $\F_2^{k'}$ with $\F_q = \F_{2^{k'}}$ in the natural ($\F_2$-linear) way.
Let $\Aconc \in \F_2^{k' \cdot k \times n' \cdot n}$ be the generator matrix of $\cC$. 
Note that $\Aconc$ can be obtained efficiently from the generator matrices of $\Cin$ and $\Cout$.

Throughout the paper, we shall use $\sigma$ to denote an outer codeword and $c$ to denote a codeword in the concatenated code $\cC$. 
To ease notation, for $c \in \cC$, we will denote $c[m] := \Cin(\sigma[m]) = c[(m-1)\cdot n' + 1:m\cdot n'] \in \F_2^{n'}$. Namely, $c[m]$ is the $m$th inner codeword inside the concatenated codeword $c$. Similarly, for a row $a$ of the generator matrix $A$, we will let $a[m] := a[(m-1)\cdot n' + 1:m\cdot n']$. (Note that for any other string in the paper, when we write $x[m]$, we mean the $m$th \emph{bit} in the string $x$; we use this notation only for codewords $c$ in the concatenated code $\cC = \Cout \circ \Cin$, including the rows of $A$).

\paragraph{Ordering the base code.}
We define an order on the codewords $c_0, c_1, \ldots, c_{2^{kk'}-1}$ of our concatenated code $\cC$.  Define $c_0$ to be the zero codeword. For $i > 0$,
The $i$th codeword in $\cC$ is defined by
\begin{equation}
    c_i= \Aconc^T \calr_{k'k}(i).
\end{equation}
As $\calr_{k'k}$ is a binary reflected code, $\calr_{k'k}(i-1)$ and $\calr_{k'k}(i)$ differ in exactly one index.  Let $z_i$ denote this index, so we have
\[ \calr_{k'k}(i-1)[z_i] \neq \calr_{k'k}(i)[z_i]\;. \] 
Denote by $a_m$ the $m$th row of $A$. Then, for every $i\in \{1, 2, \ldots, 2^{kk'} - 1\}$, we have
\begin{equation} \label{eq:sig-row}
    c_{i} = c_{i-1} \oplus a_{z_i}.
\end{equation}

This is clearly an ordering of all the codewords of $\cC$. Indeed, $\mathcal{R}_{k'k}$ is a bijection and $A$ is full-rank, so as $i$ varies in $\{0, \ldots, 2^{kk'}-1$\}$, c_i = A^T \mathcal{R}_{k'k}(i)$ varies over all the codewords in $\cC$, hitting each $c \in \cC$ exactly once.

    Note that the ordering of $\cC$ immediately implies an ordering of $\Cout$. Indeed, by the concatenation process, there is a bijection between $\cC$ and $\Cout$. Thus, the $i$th codeword $c_i$ in our concatenated codeword defines also the $i$th codeword  in the outer code.  
    We let $\sigma_i \in \Cout$ denote this outer codeword.  That is,
    \[ c_i = \sigma_i \circ \Cin = ( \Cin(\sigma_i[1]), \ldots, \Cin(\sigma_i[n]) ).\]

\subsection{Intermediate Code}
Next, we explain how to get our intermediate code $\calw$ from our base codes $\Cout, \Cin$, and $\cC = \Cout \circ \Cin$.
\paragraph{Encoding the generator matrix row difference.}
Recall that the difference of every two consecutive  codewords is a row of the generator matrix $A$, namely, $c_i - c_{i-1} = a_{z_i}$. 
In the $i$th codeword of the intermediate code $\calw$, we will include $z_i$, encoded with a repetition code that repeats each bit of $\repblow$ times. 
We shall explicitly state the value of $\rowenc$ when we prove \Cref{thm:main} and choose the parameters of our scheme.
Since $z_i$ can be represented using $\log(kk')$ 
{bits}, we shall encode $z_i$ {using the map}\footnote{Note that $\log(kk')$ might not be an integer. Going forward, we will drop floors and ceilings in order to ease notation and the analysis. We note that the loss in the rate due to these roundings is negligible and does not affect the asymptotic results.} $\call:\mathbb{F}_2^{\log(kk')} \rightarrow \mathbb{F}_2^{\rowenc}$ 
which simply {performs} repetition encoding described above, to obtain
\[ L_{z_i} = \call(z_i) \;.\]

\paragraph{Construction of $\calw$.}
Now, we describe how to generate our intermediate code $\calw$.  Informally, to get the $i$th codeword $w_i \in \calw$, we take the $i$th codeword $c_i \in \cC$; add $L_{z_i}$ at the beginning; and then break up $c$ by including short strings of repeated bits in between each inner codeword $c_i[m] \in \Cin$.  Formally, we have the following definition.
\begin{definition} 
\label{def:calw}
    Let $B$ be a (constant) integer that will be chosen later. Let $d := n' n + B (n+1)+ \rowenc$. 
    The intermediate code $\calw$, along with its ordering, is defined as follows. For each $i \in \{0, \ldots, q^k-1\}$, define $w_i \in \zo^{d}$ by the equation
    \begin{equation}
        w_i =\begin{cases}
             L_{z_i} \circ 0^{B} \circ  c_i[1]\circ 0^{B} \circ  \cdots \circ 0^{B} \circ c_i[n] \circ 0^{B} & \text{if } i \text{ is even} \\
            L_{z_i} \circ 1^{B} \circ  c_i[1]\circ 1^{B} \circ \cdots \circ 1^{B} \circ c_i[n] \circ 1^{B} & \text{if } i \text{ is odd}
        \end{cases}
    \end{equation}
where $c_i$ is the $i$th codeword in $\cC$, and where we recall that $c_i[m] = \Cin(\sigma_i[m])$ denotes the $m$th inner codeword in $c_i$.  Finally, we define $\calw \subseteq \F_2^{d}$ by
\[\calw = \{ w_i \,:\, i \in \{0, 1, \ldots, q^k-1\} \}.\] 
Note that $\calw$ has the natural ordering $w_0, w_1, \ldots, w_{q^k-1}$.
\end{definition}

\subsection{The Final Code}

To create our robust Gray code $\calg$, given any two consecutive codewords in $\calw$, we inject extra codewords between them to create $\calg$.  Before we formally define this, we begin with some notation.
\begin{definition}[The parameters $r_i, h_{i,j}, \bar{j}$]\label{def:notation}
Let $\calw \subseteq \zo^d$ be a code defined as in Definition \ref{def:calw}. For each $i \in \{0, \ldots, q^k-1\}$,  define $r_i = \sum_{\ell=1}^{i} \Delta(w_{\ell-1},w_{\ell})$, and let $N = r_{q^k-1}$. 
Also, for $i \in \{0, \ldots, q^k-1\}$ and $1\leq j < \Delta (w_i, w_{i+1})$, let $h_{i,j} \in [d]$ be the $j$th index where codewords $w_i$ and $w_{i+1}$ differ. 
We will also define $h_i =(h_{i,1}, h_{i,2}, \ldots , h_{i,\Delta (w_i, w_{i+1})-1} ) \in [d]^{\Delta(w_i, w_{i+1})-1}$ to be the vector of all indices in which $w_i$ and $w_{i+1}$ differ, in order, except for the last one.\footnote{The reason we don't include the last one is because of Definition~\ref{def:calg} below, in which we flip bits one at a time to move between the codewords $g_j$ of our robust Gray code $\calg$. In more detail, the reason is because once the last differing bit has been flipped, $g_j$ will lie in $[w_{i+1}, w_{i+2})$, not $[w_i, w_{i+1})$.} 
Finally, for $i \in \{0, \ldots, q^k-1\}$ and for $j \in [r_i, r_{i+1})$, we will use the notation $\bar{j}$ to denote $j - r_i$.  
That is, $\bar{j}$ is the index of $j$ in the block $[r_i, r_{i+1})$ in which $j$ falls.
\end{definition}

With this notation, we are ready to define our robust Gray code $\calg$.
\begin{definition}[Definition of $\calg$]
\label{def:calg}
Define the zero'th codeword of $\calg$ as $g_0 = w_0$. 
Fix $j \in \{1, \ldots, N-1\}$.  If  $j = r_i$ for some $i$, we define $g_j \in \zo^d$ by $g_j = w_i$.
On the other hand, if 
$j \in (r_i, r_{i+1})$ for some $i$, then we define $g_j \in \zo^{d}$ as
\begin{equation}
g_j = \text{pref}_{h_{i,\bar{j}}}(w_{i+1}) \circ \text{suff}_{h_{i,\bar{j}}+1}(w_i).
\end{equation}
Finally, define $\calg \subseteq \zo^d$ by $\calg = \{g_j \,:\, j \in \{0, \ldots, N-1\}\},$ along with the encoding map $\enc{\calg}: \{0,\ldots, N-1\} \to \zo^d$ given by $\enc{\calg}(j) = g_j$. 
\end{definition}

Note that when $j \in [r_i, r_{i+1})$, the last bit that has been flipped to arrive at $g_j$ in the ordering of $\calg$ (that is, the ``crossover point'' alluded to in the introduction)  is $h_{i,\bar{j}}.$
We make a few useful observations about Definition~\ref{def:calg}.  The first observation follows immediately from the definition.  
\begin{observation}[$\calg$ is a Gray code]\label{obs:graycode}
$\calg$ is a Gray code.  That is,
    for any $j \in \{0, \ldots, N-1\}$, we have that $ \Delta( g_j , g_{j+1} ) = 1$.  
\end{observation}

Next, we bound the rate of $\calg$.
\begin{observation}[Rate of $\calg$]\label{obs:rate} 
The rate of the robust Gray code $\calg$ defined in \Cref{def:calg} is at least 
\begin{equation}\label{eq:calg-rate}
\frac{\Rout \Rin}{1 + \frac{B}{n'}\cdot(1 + \frac{1}{n})  + \frac{L}{nn'}}.
\end{equation}
\end{observation}
\begin{proof}
Recall that $\Cin$ has rate $\Rin$ and $\Cout$ has rate $\Rout$. 
Then the code $\calw$ constructed as in Definition~\ref{def:calw} has rate 
\begin{align*}
    \frac{\log q^k}{n' \cdot n + B(n + 1) + \rowenc} &= \frac{\Rout n \cdot \Rin n'}{n' \cdot n + B(n + 1) + \rowenc} \\
    & = \frac{\Rout \Rin}{1 + \frac{B}{n'}\cdot (1 + \frac{1}{n}) + \frac{\rowenc}{ nn'}} \;.
\end{align*}
Thus, the rate of $\calg$ is at least the above, given that $\calg$ has more codewords than $\calw$ but the same length.
\end{proof}

Fix $i$, and suppose that $g_j$ is obtained as an intermediate codeword between $w_i$ and $w_{i+1}$.  Then, on the coordinates in which $w_i$ and $w_{i+1}$ differ, $g_j$ will disagree with $w_{i}$ for a first chunk of them, and agree with $w_i$ for the rest.  We make this precise in the following observation.
\begin{observation}\label{obs:unary}
    Let $g_j \in \calg$, and suppose that $j \in (r_i, r_{i+1})$ for some $i \in \{0, \ldots, q^k -1\}$.  
    Recall from \Cref{def:notation} that $h_i \in [d]^{\Delta(w_i, w_{i+1})-1}$ is the vector of positions on which $w_i$ and $w_{i+1}$ differ (except the last one).
    Then
    \[ (g_j + w_i)[h_i] = \enc{\calu}(\bar{j}),\]
    where $\calu \subset \zo^{\Delta(w_i, w_{i+1})-1}$ is the unary code of length $\Delta(w_i, w_{i+1})-1$.  Above, $(g_j + w_i)[h_i]$ denotes the restriction of the vector $g_i + w_i \in \F_2^d$ to the indices that appear in the vector $h_i$. 
    
     Further, for every $m\geq \bar{j}$, we have 
    \[
    (g_j + w_i)[h_i[1:m]] = \enc{\calu}(\bar{j}),
    \]
    where $\calu$ is the unary code of length $m$. That is, even if we take the first $m$ values of $h_i$, then as long as $m \geq \bar{j}$, the restriction of $(g_j + w_i)$ to these values  match the unary encoding of $\bar{j}$.
\end{observation}
\begin{proof}
    By definition, $h_i$ contains the indices on which $w_i$ and $w_{i+1}$ differ, and also by definition, by the time we have reached $g_j$, the first $j - r_i = \bar{j}$ of these indices have been flipped from agreeing with $w_i$ to agreeing with $w_{i+1}$.  Thus, if we add $g_j$ and $w_i$ (mod 2), we will get $1$ on the first $j-r_i$ indices and $0$ on the on the rest.  The ``further'' part follows immediately.
\end{proof}

Our next objective is to show that Definition~\ref{def:calg} actually defines an injective map.
We begin by providing some notation for different parts of the codeword $g_j\in \calg$. For a string $x$, $x[m:m']$ denotes the substring $(x_m, x_{m+1} , \ldots , x_{m'})$. For any $x\in \zo^d$ define 
\begin{itemize}
    \item $\gind = x[1:\rowenc ] $,
    \item $s_m =  x [ \rowenc + (m-1) (B +n') + 1 :  \rowenc + (m-1) (B +n') + B]  $ for $m \in [n + 1 ]$,
    \item $\tilde{c}_m = x[ \rowenc + m B  + (m-1) n' +1 : \rowenc + m (B +n')  ]$ for $m\in[n] $,
\end{itemize}
As a result, any string $x$ that is either a codeword or a corrupted codeword, has the following format (see also Figure~\ref{fig:notation}): 
\begin{equation}
    x = \gind \circ s_1 \circ \tilde{c}_1 \circ \ldots \tilde{c}_n \circ s_{n+1} \;.
\end{equation} 
 We will call each of $\gind_i, s_i, c_i$ a \textbf{\chnk}. For a codeword $g_j\in \calg $ such that $j\in [r_i, r_{i+1})$ we will call a \chnk\  a \textbf{\chunkn} if it is equal to its corresponding \chnk\ in either $w_i$ or $w_{i+1}$.  
 This notation is illustrated in Figure~\ref{fig:notation}.

 \begin{figure}
 \begin{center}
     \begin{tikzpicture}[xscale=1.1]
         \draw (0,0) rectangle (12.5,1);
         \node at (-.5,.5) {$x = $};
         \draw[fill=red!10] (0,0) rectangle (1.5,1);
         \node at (.75,.5) {$\gind$};
    
         \draw[fill=blue!20] (1.5, 0) rectangle (2,1);
         \node at (1.75, .5) {$s_1$};
         \draw[fill=yellow!20] (2,0) rectangle (3,1);
         \node at (2.5, .5) {$\tilde{c}_1$};
         \draw[fill=blue!20] (3, 0) rectangle (3.5,1);
         \node at (3.25, .5) {$s_2$};
         \draw[fill=yellow!20] (3.5,0) rectangle (4.5,1);
         \node at (4, .5) {$\tilde{c}_2$};
         \draw[fill=blue!20] (4.5, 0) rectangle (5,1);
         \node at (4.75, .5) {$s_3$};
         \draw[fill=yellow!20] (5,0) rectangle (6,1);
         \node at (5.5, .5) {$\tilde{c}_3$};
      
         \draw[fill=blue!20] (12, 0) rectangle (12.5,1);
         \node[rotate=90] at (12.25, .5) {$s_{n+1}$};
            \begin{scope}[xshift=-.5cm]
         \draw[fill=blue!20] (11, 0) rectangle (11.5,1);
         \node at (11.25, .5) {$s_n$};
         \draw[fill=yellow!20] (11.5,0) rectangle (12.5,1);
         \node at (12, .5) {$\tilde{c}_n$};
         \draw[fill=blue!20] (9.5, 0) rectangle (10,1);
         \node[rotate=90] at (9.75, .5) {$s_{n-1}$};
         \draw[fill=yellow!20] (10,0) rectangle (11,1);
         \node at (10.5, .5) {$\tilde{c}_{n-1}$};
         \draw[fill=yellow!20] (8.5,0) rectangle (9.5,1);
         \node at (9, .5) {$\tilde{c}_{n-2}$};
         \node at (7.25,.5) {$\cdots$};\end{scope}
              \draw[ultra thick,red,rounded corners] (0,-.5) rectangle (1.5, 1.5);
         \node[red](a) at (.75,3) {\begin{minipage}{2cm} \begin{center} This is a \chnk.\end{center} \end{minipage}};
         \draw[->,red,thick] (a) to (.75, 1.5);
         \draw[ultra thick,violet,rounded corners] (1.5,-.5) rectangle (2, 1.5);
         \node[violet](b) at (1.75,2.3) {\begin{minipage}{1cm} \begin{center} So is this.\end{center} \end{minipage}};
         \draw[->,violet,thick](b) to (1.75,1.5);
         \node[green!50!black] at (2.5, 1.75) {etc.};

    \begin{scope}[yshift=-3cm]
         \draw (0,0) rectangle (12.5,1);
         \node at (-.7,.5) {$g_j = $};
         \draw[pattern=north west lines, pattern color=orange!60] (0,0) rectangle (1.5,1);
         \node at (.75,.5) {$\gind$};
         \draw[pattern=north west lines, pattern color=cyan!20] (1.5, 0) rectangle (2,1);
         \node at (1.75, .5) {$s_1$};
         \draw[pattern=north west lines, pattern color=green!40] (2,0) rectangle (3,1);
         \node at (2.5, .5) {$\tilde{c}_1$};
         \draw[pattern=north west lines, pattern color=cyan!20] (3, 0) rectangle (3.5,1);
         \node at (3.25, .5) {$s_2$};
         \draw[pattern=north west lines, pattern color=green!40] 
         (3.5,0) rectangle (4,1);
         \draw[fill=yellow!20] (4,0) rectangle (4.5, 1);
       
         \draw[fill=blue!20] (4.5, 0) rectangle (5,1);
         \node at (4.75, .5) {$s_3$};
         \draw[fill=yellow!20] (5,0) rectangle (6,1);
         \node at (5.5, .5) {$\tilde{c}_3$};
           \draw[fill=blue!20] (12, 0) rectangle (12.5,1);
         \node[rotate=90] at (12.25, .5) {$s_{n+1}$};
         \begin{scope}[xshift=-.5cm]
         \draw[fill=blue!20] (11, 0) rectangle (11.5,1);
         \node at (11.25, .5) {$s_n$};
         \draw[fill=yellow!20] (11.5,0) rectangle (12.5,1);
         \node at (12, .5) {$\tilde{c}_{n}$};
         \draw[fill=blue!20] (9.5, 0) rectangle (10,1);
         \node[rotate=90] at (9.75, .5) {$s_{n-1}$};
         \draw[fill=yellow!20] (10,0) rectangle (11,1);
         \node at (10.5, .5) {$\tilde{c}_{n-1}$};
         \draw[fill=yellow!20] (8.5,0) rectangle (9.5,1);
         \node at (9, .5) {$\tilde{c}_{n-2}$};
         \node at (7.5,.5) {$\cdots$};
         \end{scope}
        \node[red](h) at (4, 1.5) {$h_{i,\bar{j}}$};
        \draw[red, ultra thick] (h) to (4, -.1);
        \draw[violet,rounded corners, ultra thick] (3.5,-.5) rectangle (4.5,2);
        \node[violet](a) at (9,-1) {This is \emph{not} a full chunk.  All the other chunks are full chunks.};
        \draw[->,violet, thick] (a) to[out=180,in=-90] (4,-.55);
        \end{scope}
     \end{tikzpicture}
 \end{center}
 \caption{The notation used to break up vectors $x \in \{0,1\}^d$ into {\chnk}s (top), and the distinction between {\chnk}s and {\chunkn}s when $x$ happens to be a codeword $g_j$ (bottom).  Notice that for $g_j$, if $j \in [r_i, r_{i+1})$ then we have, e.g., $s_m = b_i^{B}$ and $\tilde{c}_m = c_i[m]$ or $c_{i+1}[m]$, whenever the corresponding {\chnk}s are {\chunkn}s.}\label{fig:notation}
 \end{figure}
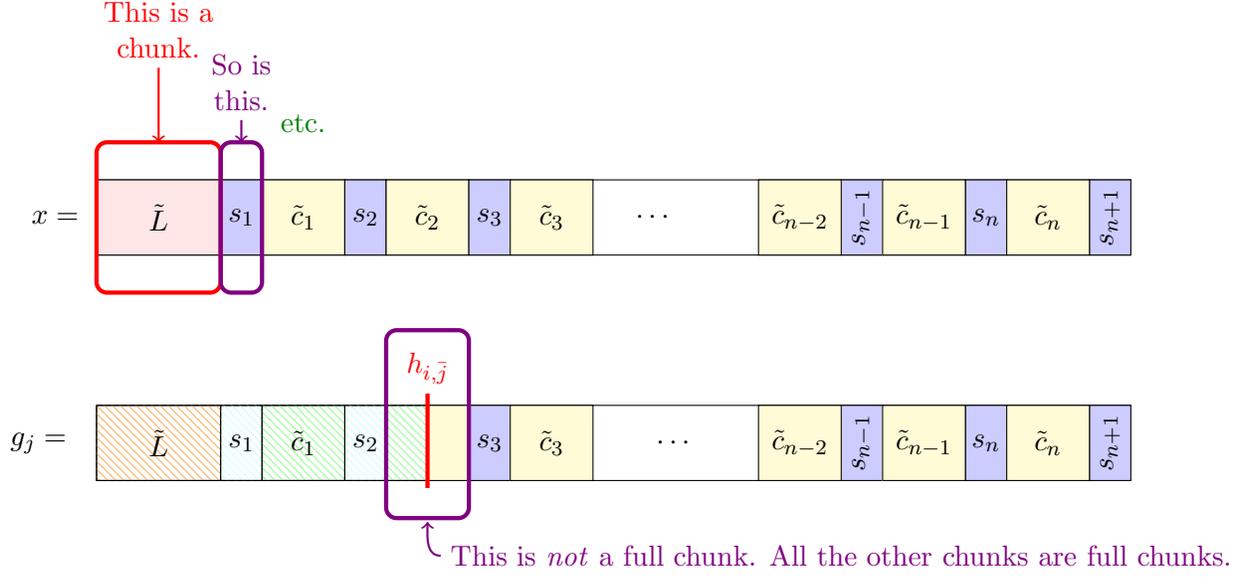
 
 The following lemma shows that for each $g_j$ there is at most a single \chnk\ that is not a \chunkn.

\begin{lemma}\label{lemma:one-broken-chunk}
    Fix $j \in \{0,\ldots,N-1\}$.  Suppose that $i \in \{0, \ldots, q^k - 1\}$ is such that $j \in [r_i, r_{i+1})$, so 
    $g_j\in \calg $ can be written as $ g_j =  \gind \circ s_1 \circ c_1 \circ \ldots c_n \circ s_{n+1} $ as above.   
    Then at most one of the substrings in  $ \cals = \sett{\gind, s_1 \ldots, s_{n+1}, c_1, \ldots, c_n }$  is not equal to the corresponding substring in  $w_i$ or $w_{i+1}$.
\end{lemma}
\begin{proof}
    First, suppose that  $j = r_i$.  Then in that case $g_j = w_i$ and all of the substrings in $\mathcal{S}$ are equal to their corresponding substring.  Otherwise, $ j \in (r_i, r_{i+1})$. In that case, $\bar{j} \in [1, r_{i+1} - r_i) = [1, \Delta(w_i, w_{i+1}))$. 
    This means that $h_{i,\bar{j}}$ (the ``crossover point'' for $g_j$) is defined, and indexes a position in $g_j$, and in particular in one of the sub-strings in $\mathcal{S}$.  Then other substrings strictly to the left of $h_{i,\bar{j}} $ are equal to their corresponding substring in $w_{i+1}$; and the ones  strictly to the right are equal to the corresponding substring in $w_i$.
\end{proof}

Next, we show that there are no ``collisions'' in $\calg$; that is, there are no $j \neq j'$ so that $g_j = g_{j'}$.
\begin{lemma}\label{lem:injective}
    Let $\calg$ and $\enc{\calg}$ be as in \Cref{def:calg}. Then $\enc{\calg}$ is injective.  
\end{lemma}

\begin{proof} 
   Assume, for the sake of contradiction, that there are two distinct $j,j^\prime \in \sett{0, \ldots , N-1}$ such that $g_j= g_{j^\prime}$.  Without loss of generality assume that $ j^\prime > j$.  There are three scenarios possible. 
    \begin{enumerate}
        \item  \textbf{Case 1:} Both $j$ and $j^\prime$ are in the interval $[r_i, r_{i+1})$.  Then we claim that $g_j[h_{i,\bar{j'}}] \neq g_{j^\prime}[h_{i,\bar{j'}}]$. The reason is that  $g_j[h_{i,\bar{j'}}] = w_i[h_{i,\bar{j'}}]$ and $g_{j^\prime} [h_{i,\bar{j'}}] = w_{i+1} [h_{i,\bar{j'}}]$. This implies that $w_i[h_{i,\bar{j'}}] = w_{i+1}[h_{i,\bar{j'}}]$ which contradicts the definition of $h_{i,\bar{j'}}$.

        \item \textbf{Case 2: }$j\in [r_{i-1}, r_i)$ and $j^\prime \in [r_{i}, r_{i+1})$. 
            Then $g_j$ is an interpolation of $ w_{i-1}$ and  $w_i$, and $g_{j^\prime}$ is an interpolation of $ w_{i} $ and $w_{i+1}$.
            Denote 
            \[
            g_j = \gind \circ s_1 \circ \tilde{c}_1 \circ \ldots \tilde{c}_n \circ s_{n+1} 
            \] 
            and 
            \[
            g_{j'} = \gind'_1 \circ s'_1 \circ \tilde{c'}_1 \circ \ldots \tilde{c'}_n \circ s'_{n+1}  \;.
            \]

            If $h_{i,\bar{j}}$ does not fall into $\gind$ or $s_{n+1}$ then $(s_1, \ldots, s_{n}, s_{n+1})$ cannot be equal to $(s'_1,\ldots, s'_{n}, s'_{n+1})$. Indeed,  assuming without loss of generality that $i$ is even, then $(s_1, \ldots, s_{n+1}) = 0^a1^b$ where both $a$ and $b$ are nonzero, while $(s'_1,\ldots, s'_{n}, s'_{n+1})$ is of the form $1^{a'}0^{b'}$.
            An identical argument shows that if $h_{i+1,\bar{j'}}$ does not fall into $f'_1$ or $s'_{n+1} $ then $(s_1, \ldots, s_{n}, s_{n+1})$ cannot be equal to $(s'_1,\ldots, s'_{n}, s'_{n+1})$.
            We are left with the case where $h_{i,\bar{j}}$ falls in $\gind$ or $s_{n+1}$ and $h_{i+1,\bar{j'}}$ falls in $f'_1$ or $s'_{n+1}$. 
            In this case, since the parities of $i-1$ and $i$ are different, the only possibility to get equality between $(s_1, \ldots, s_{n+1})$ and $(s'_1, \ldots, s'_{n+1})$ is if $h_{i,\bar{j}}$ is in $\gind$ and $h_{i+1,\bar{j'}}$ falls exactly on the last bit of $s'_{n+1}$.
            This implies that $(\tilde{c}_1, \ldots, \tilde{c}_n)$---which corresponds to the outer codeword $\sigma_{i-1}$---and $(\tilde{c'}_1, \ldots, \tilde{c'}_n)$---which corresponds to the outer codeword $\sigma_{i+1}$---are equal, a contradiction of the fact that the codewords in our ordering of the outer code are all distinct.

        \item \textbf{Case 3:} $j\in [r_{i}, r_{i+1})$ and $j^\prime \in [r_{i^\prime}, r_{i^\prime+1})$  where $ |i-i^\prime| > 1$.
        As before, denote 
        \[
        g_j = \gind \circ s_1 \circ \tilde{c}_1 \circ \ldots \tilde{c}_n \circ s_{n+1}
        \] 
        and 
        \[
        g_{j'} = \gind'_1 \circ s'_1 \circ \tilde{c'}_1 \circ \ldots \tilde{c'}_n \circ s'_{n+1} \;.
        \]
        By \Cref{lemma:one-broken-chunk}, only a single chunk in $g_j$ (resp. $g_{j'}$) is not equal to the corresponding chuck in $w_i$ or $w_{i+1}$ (resp. $w_{i'}$ or $w_{i'+1}$). We shall consider several sub-cases depending on the locations of $h_{i,\bar{j}}$ and $h_{i',\bar{j'}}$.

        First, assume that $h_{i,\bar{j}}$ falls into $s_m\circ \tilde{c}_m$ and $h_{i',\bar{j'}}$ into $s'_{m'}\circ \tilde{c'}_{m'}$ where $m \neq m' \in [n]$.  
        Also, assume without loss of generality that $m' > m$. Note that since neither $h_{i,\bar{j}}$ nor $h_{i',\bar{j'}}$ fall in the last chunks ($s_{n+1}$ and $s_{n+1}'$, respectively), it must be that $i$ and $i'$ have the same parity; 
        otherwise the chunks $s_{n+1}$ and $s'_{n+1}$ would disagree, contradicting our assumption that $g_j = g_{j'}$. 
        Assume that $(s_1, \ldots, s_{n+1})$ and $(s'_1, \ldots, s'_{n+1})$ are of the form $1^a0^b$ and $1^{a'}0^{b'}$, respectively. Clearly, as $h_{i,\bar{j}}$ falls into $s_m\circ \tilde{c}_m$ and $h_{i',\bar{j'}}$ into $s'_{m'}\circ \tilde{c'}_{m'}$, and $m'>m$, it must be that $a' > a$. We conclude that $g_j \neq g_{j'}$, a contradiction.

        Now assume that both $h_{i,\bar{j}}$ and $h_{i',\bar{j'}}$ fall into $s_m\circ \tilde{c}_m$ and $s'_{m}\circ \tilde{c'}_{m}$, respectively.  (Note that the difference between this sub-case and the previous one is that $h_{i,\bar{j}}$ and $h_{i', \bar{j'}}$ fall into chunks with the \emph{same} index $m$). 
        In this case, since $\gind$ and $f'_1$ in $g_j$ and $g_{j'}$ are full chunks, it holds that the tuple
        \[
        (L_{z_{i+1}}, c_{i+1}[1], \ldots, c_{i+1}[m-1],c_{i}[m+1],\ldots, c_{i}[n])
        \]
        is equal to
        \[ 
        (L_{z_{i'+1}}, c_{i'+1}[1], \ldots, c_{i'+1}[m-1],c_{i'}[m+1],\ldots, c_{i'}[n])\;.
        \]
        Now, since $L_{z_{i+1}} = L_{z_{i'+1}}$, we have that $c_{i+1}$ and $c_{i'+1}$ are obtained by adding the same row $a_z$ of the generator matrix $A$, to $c_i$ and $c_{i'}$, respectively. Thus, for each $r \leq m-1$ we have that $c_i[r] = c_{i'}[r]$ and in total,
        \[
        (c_i[1], \ldots, c_i[m-1], c_i[m+1], \ldots, c_i[n]) = (c_{i'}[1], \ldots, c_{i'}[m-1], c_{i'}[m+1], \ldots, c_{i'}[n])
        \]
        which contradicts the fact that $i\neq i'$ and that the minimum distance of $\Cout$ satisfies $\delout n > 1$.
        
        Finally, we consider the sub0case where $h_{i,\bar{j}}$ falls in $\gind$ or $s_{n+1}$. 
        In this case, if $h_{i',\bar{j'}}$ also falls in $f'_1$ or $s'_{n+1}$, then $(\tilde{c}_1, \ldots, \tilde{c}_n)$ and $(\tilde{c'}_1, \ldots, \tilde{c'}_n)$ correspond to two distinct outer codewords, which implies that $g_j \neq g_{j'}$, contradicting our assumption that they are the same.
        If $h_{i',\bar{j'}}$ doesn't fall in $f'_1$ or $s'_{n+1}$, then it must fall into an $s'_m$ or $\tilde{c'}_m$ for some $m\in [n]$. In this case, $(s_1, \ldots, s_n, s_{n+1})$ will be the all $1$ or all $0$ string but $(s'_1, \ldots, s'_n, s'_{n+1})$ clearly cannot be the all $0$ or $1$ string since $s'_{n+1}\neq s'_{m}$.
    \end{enumerate}
    Thus, in all cases we arrive at a contradiction, and this completes the proof.
\end{proof}

\section{Decoding Algorithm}\label{sec:decalg}
In this section, we define the decoding algorithm. In the following paragraphs, we will give a high level overview of the major steps in the decoding procedures. We denote the input to the algorithm by $x \in \F_2^d$, and we recall that $x$ is of the following form (see also Figure~\ref{fig:notation}):
\[
x = \gind \circ s_1 \circ \tilde{c}_1 \circ \ldots \circ \tilde{c}_n \circ s_{n+1} \;.
\]
Recall that $x$ is a noisy version of some codeword of $\calg$; let us write $x = g_{j} \oplus \eta$ for a noise vector $\eta \in \F_2^d$, so our objective is to return $\hat{j} \approx j$. As usual, suppose that $j \in [r_i, r_{i+1})$, and define $\bar{j} = j -r_i$, so that $h_{i,\bar{j}}$ is the crossover point in the correct codeword $g_j$.

Our primary decoding algorithm, $\dec{\calg}$, is given in Algorithm~\ref{alg:Dec}. 
The first objective of the decoding algorithm is to estimate the chunk in which the crossover point $h_{i,\bar{j}}$ occurs. We define $\ell \in \{0,\ldots, n+1\}$ to be 
   \begin{equation}
   \ell = 
   \begin{cases}
        0& \text{ if $h_{i,\bar{j}}$ falls in $\gind$}\\
        m& \text{ if $h_{i,\bar{j}}$ falls in $s_m \circ \tilde{c}_m$ for $m\in [n]$}\\
        n+1& \text{ if $h_{i,\bar{j}}$ falls in $s_{n+1}$}
   \end{cases} \;. \label{eq:ell-def}
   \end{equation}
    Intuitively speaking, $\ell$ will be the crossover point at the level of chunks. Algorithm~\ref{alg:Dec} will estimate $\ell$, and we will denote this estimation by $\elhat$.

Next, Algorithm~\ref{alg:Dec} decodes each chunk $\tilde{c}_m$ using the inner code's decoding algorithm to obtain an estimate $\sighat[m] \in \F_q$. 
    Then,
    based on the location of $\elhat$ and the decoded symbols $\sighat[m]$, we either invoke Algorithm~\ref{alg:est} ($\getest$), or Algorithm~\ref{alg:best} ($\bgetest$) in order to obtain our final estimate~$\jhat$.

    In more detail, for an appropriate constant $\gapp \in (0,1)$, we will show that with high probability, $\ell$ cannot be more than $\beta n$ ``far'' from $\elhat$.  
    We break up both our algorithm and analysis into two cases, depending on whether $\elhat$ lands in $(\beta n, n - \beta n)$.

    If $\elhat \in (\gapp n, n - \gapp n)$, we say that $\elhat$ is \emph{in the middle}.  In this case, we call Algorithm~\ref{alg:est} to recover $\jhat$.
    If $\elhat \not\in (\gapp n, n- \gapp n)$, we say that $\elhat$ is \emph{in the boundaries}.  In this case,  we call Algorithm~\ref{alg:best} to recover $\jhat$.
    We next describe Algorithm~\ref{alg:est} and Algorithm~\ref{alg:best}, and why we break things into these two cases.
    
    Algorithm~\ref{alg:est} ($\getest$) is called when $\elhat \in (\gapp n, n - \gapp n)$. The first thing it does is to update our estimate $\sighat$---which corresponds to an interpolation between two codewords of $\Cout$---to obtain a version $\sighat$ that corresponds to only one codeword in $\Cout$.  To do this, it first decodes the first $\rowenc$ bits to get $z_{i+1}$ and uses this to update $\sighat$ by:
    \[
        \sighat[m] = 
        \begin{cases}
            \sighat[m] -  a_{z_{i+1}}[m]& \text{ if $m < \elhs$} \\
            \perp& \text{ if $m \in [\elhs, \elhe]$} \\
            \sighat[m]& \text{ if $m > \elhe$}
        \end{cases}\;.
    \]
    where we use the $\perp$ symbol to indicate an erasure.
    Above, $\elhs = \elhat - \gapp n$, $\elhe = \elhat + \gapp n$ and recall that $a_{z_{i+1}}[m] = a_{z_{i+1}}[(m-1)n' +1 : mn']$.  Also, above we have used the fact that $a_{z_{i+1}}[m] \in \Cin$, and thus corresponds to some element of $\F_q$, so we treat $a_{z_{i+1}}[m]$ as an element of $\F_q$ in the subtraction above.  Intuitively, what the algorithm is doing here is translating the elements of $\sighat$ that correspond to $c_{i+1}$ into elements that correspond to $c_i$. Finally, \Cref{alg:est} uses $\Cout$'s decoder on $\sighat \in \F_q^n$ to obtain $\ihat$.  Given $\ihat$, it computes $\jhat$ by taking into consideration how many bits were flipped from $w_{\ihat}[H]$ to get $x[H]$, where $H = \{i\mid w_{\ihat} \neq w_{\ihat +1}\}$.

    Algorithm~\ref{alg:best} ($\bgetest$) is invoked when $\elhat \notin (\gapp n, n - \gapp n)$. The general strategy in this algorithm is similar to that of Algorithm~\ref{alg:est}, but there are several differences. 
    The main reason for these differences is that if $\elhat$ is in the boundaries, $\elhat$ will only be ``close'' to $\ell$ modulo $n$.  To see intuitively why this should be true,  consider two scenarios, one where $j$ is all the way at the \emph{end} of the interval $[r_i, r_{i+1})$, and a second where $j$ is all the way at the \emph{beginning} of the next interval $[r_{i+1}, r_{i+2})$.  The $j$'s in these two scenarios are close to each other, and their corresponding encodings under $\calg$ are also close in Hamming distance.  However, in the first scenario, $\ell$ is close to $n+1$, while in the second scenario, $\ell$ is close to $0$.  Thus, we should only expect to be able to estimate $\ell$ modulo $n$, and it could be possible that, for example, $\elhat$ is close to zero while $\ell$ is close to $n$.

    Here is how we take this into account in Algorithm~\ref{alg:best}, relative to Algorithm~\ref{alg:est} discussed above.
    First, we define $\elhs$ and $\elhe$ slightly differently, taking them modulo $n$ as per the intuition above (see Figure~\ref{fig:ell}).
    Second, Algorithm~\ref{alg:best} sets $\sighat_i[m]$ differently. For $m \in [1, \elhs] \cup [\elhe, n]$ we set $\sighat_i[m] = \perp$. A crucial observation is that for every $m \in [\elhe, \elhs]$,  
    if $\elhat \leq \gapp n$, then $\tilde{c}_m$ is a corrupted version of $c_i[m]$ and if $\elhat \geq n - \gapp n$ then $\tilde{c}_m$ is a corrupted version of $c_{i+1}[m]$. 
    Since we could have \emph{either} $\ell \leq \beta n$ \emph{or} $\ell \geq n - \beta n$, we thus take \emph{both} of these cases into account, and consider both $c_i$ and $c_{i+1}$ as possibilities.  To this end, we compute \emph{two} possible decodings of $\sighat$, and we then get \emph{two} options for $\jhat$, call them $\jhat_1$ and $\jhat_2$, by performing the same steps as in Algorithm~\ref{alg:est} to each case. Then \Cref{alg:best} sets $\jhat$ to be the more likely of $\jhat_1$ and $\jhat_2$.

    Finally, we discuss our last helper function, Algorithm~\ref{alg:compr}, called $\compr$.  This helper function is called in both Algorithms~\ref{alg:est} and \ref{alg:best}.  Its job is to compute $r_i$ given $i$.  While this seems like it should be straightforward---after all, $r_i$ is defined in Definition~\ref{def:notation} as a simple function of $i$---doing this efficiently without storing a lookup table of size $q^k$ requires some subtlety.  The key insight---and the reason that we defined the order on $\cC$ the way we did---is that from \eqref{eq:sig-row}, we have
    \[ 
    c_i = c_{i-1} \oplus a_{z_i},
    \]
    where we recall that $z_i$ is the index in which $\mathcal{R}_k(i)$ and $\mathcal{R}_k(i-1)$ differ, and $a_{z_i}$ is the $z_i$'th row of the generator matrix $A$ of $\cC$.  To see why this matters, recall from Definition~\ref{def:notation} that
    \begin{align}\label{dec:delta1}
     r_i = \sum_{t=1}^i \Delta( w_{t-1}, w_t ).
     \end{align}
    There are contributions to each $\Delta(w_{t-1}, w_t)$ from each of the chunks $\gind$, $s_m$, and $\tilde{c}_m$.  Here, we discuss just the contribution from the $\tilde{c}_m$ chunks, as this illustrates the main idea.  Due to \eqref{eq:sig-row}, this contribution is
 \begin{align}\label{dec:delta2}
 \sum_{t=1}^i \Delta( \sigma_{t-1} \circ \Cin , \sigma_t \circ \Cin) =  \sum_{t=1}^i \|a_{z_t}\|.
 \end{align}
We cannot afford to add up all of the terms in the sum individually, as $i$ may be as large as $q^k$.  However, instead we can compute the number of times that a particular row $a_z$ appears in the sum above (this is given by Observation~\ref{obs:BRC}), and add $\|a_{z}\|$ that many times.  As there are only $k\cdot k'$ such rows, this can be done efficiently.

This wraps up our informal description the decoding algorithm $\dec{\calg}$ and its helper functions; we refer the reader to the pseudocode for formal descriptions.
    In the next section, we present the analysis of $\dec{\calg}$.

\begin{algorithm}[h!]
\caption{$\dec{\calg}$: Decoding algorithm for $\calg$}\label{alg:Dec}
\begin{algorithmic}[1] 
    \State \textbf{Input:} $x \in \F_2^{d}$
      \colorcomment{ Estimate location of broken chunk: }
    \For{ $ m \in \{1, \ldots, n+1\}$  } 
    \State $\shat_m = \maj(s_m)$ \label{line4:maj-dec}
    \EndFor
    \State $\shat = (\shat_1, \ldots, \shat_{n+1})$
    \State $\elhat_1 = \dec{\calu}(\shat)$ \label{line4:lhat-start}
    \State $\elhat_2 = \dec{\calu^{\text{comp}}} (\shat) $
   \State $\elhat = \begin{cases} \elhat_1 & \Delta(\shat, 1^{\elhat_1} 0^{n + 1 - \elhat_1})) < \Delta (\shat, 0^{\elhat_2}1^{n+1-\elhat_2}) \\ \elhat_2 & \text{else} \end{cases}$\label{line4:lhat-comp}
    \colorcomment{Decode inner code $\Cin$:}
    \For{ $ m \in \{ 1 , \ldots , n  \}$} \label{line4:dec-inn-code-s}
    \State $\sighat[m] = \dec{\Cin} ( \tc[m] )$
        \EndFor\label{line4:dec-inn-code-e}
    \colorcomment{Estimate $j$:}
    \If{$\elhat \in (\gapp n, n - \gapp n)$}
    \State $ \jhat = \getest (x, \sighat , \elhat)$
    \label{line4:getest}
    \Else
    \State $\jhat=  \bgetest (x, \sighat , \elhat)$\label{line4:bgetest}
    \EndIf
    
    \State \textbf{Output:} $\jhat$

\end{algorithmic}
\end{algorithm}

\begin{algorithm}[H]
\caption{ $\getest $: Computing the final estimate of $\jhat$}\label{alg:est}
\begin{algorithmic}[1] 
\Require $x\in \F_2^d,\sighat \in\F_q^n, \elhat \in \{0,1,\ldots, n+1\}$ 
\colorcomment{Calculate erasure interval:}
\State $\elhs = \elhat - \gapp n$ 
\State $\elhe = \elhat + \gapp n$
\colorcomment{Update $\sighat$, taking into account the estimate of the crossover point:}

\State $\zhat = \dec{\call} ( \tilde{L} )  $ \label{line3:decode-alpha-z}

\For{$m < \elhs $} \label{line3:set-corrupt-outer-code}
\State $\sighat[m] = \sighat[m] - a_{\zhat} [m] $  \label{line3:reverse}

\gComment{$a_{\zhat}[m] \in \F_2^{n'}$ corresponds to an elt. of $\Cin$ and hence of $\F_q.$  Here, we treat $a_{\zhat}[m] \in \F_q$. }

\EndFor
\For {$ m \geq \elhs $ and $m \leq  \elhe$}
    \State $\sighat[m] = \bot$ \gComment{Set $\hat{\sigma}[m]$ to an erasure if $m$ is close to $\elhat$.}
\EndFor
\For {$m > \elhe$}
\State $\sighat[m] = \sighat[m]$  \gComment{Don't update $\sighat[m]$.}
\EndFor \label{line:doneupdating}
\colorcomment{Decode outer code to obtain $\ihat$:}
\State $\ihat = \dec{\Cout} ( \sighat)$\label{line3:dec-out} 
\colorcomment{Compute $\jbar$ and final estimate $\jhat$:}

\State $H = \{m \,|\, w_{\ihat}[m]  \neq w_{\ihat +1 }[m]\} $ 
\State $ \hat{ \jbar} =\dec{\calu} (  x[H] \oplus w_{\ihat}[H] )  $ \label{line3:fin-unary} 
\State $ \jhat = \compr (\ihat) +\hat{ \jbar}$\label{line3:jhat-est}
\State \Return $\jhat$
\end{algorithmic}
\end{algorithm}

\begin{algorithm}[H]
\caption{ $\bgetest $: Computing the final estimate of $\jhat$ in the case where $\elhat$ lies in the boundary}\label{alg:best}
\begin{algorithmic}[1] 
\Require $ x\in \F_2^d,\sighat \in\F_q^n, \elhat \in \{0,1,\ldots, n+1\}$ 
\colorcomment{Calculate erasure interval:}
\If{$\elhat \leq \gapp n$} \label{line2:set-elhs-elhe-start}
\State $\elhe = \elhat + \gapp n$
\State $\elhs = n+1 + (\elhat - \gapp n) $
\Else
\State $\elhe = \elhat + \gapp n- (n+1) $
\State $\elhs = \elhat - \gapp n$
\EndIf \label{line2:set-elhs-elhe-end}

\colorcomment{Erase symbols too near the boundary:}
\For{$m\in [0, \elhe] \cup [\elhs,n] $} \label{line2:set-corrupt-outer-code}
\State $\sighat[m] = \bot$
\EndFor
\colorcomment{Decode outer code to obtain $\ihat$:}
\State $\ihat = \dec{\Cout} (\sighat) $ \label{line2:dec-out} 
\colorcomment{Case 1: $\ell$ is in the beginning:}
\State $H = \{m\ |\ w_{\ihat}[m]  \neq w_{\ihat +1 }[m] \text{ and }  m < \rowenc + 2 \gapp n ( n' + B ) \} $  
\State $ \hat{ \jbar}_1 =\dec{\calu} (  x[H] \oplus w_{\hat{i}}[H] )  $ 
\State $ \jhat_1 = \compr (\ihat) + \hat{\jbar}_1$\label{line2:j1} 
\colorcomment{Case 2: $\ell$ is towards the end:}
\State $H = \{m \ |\ w_{\ihat}[m]  \neq w_{\ihat -1}[m] ,  m \geq d - 2 \gapp n ( n' + B ) \} $
\State $ \hat{ \jbar}_2 =\dec{\calu^{\text{comp}}} (  x[H] \oplus w_{\hat{i}}[H] )  $ 
\State $ \jhat_2 = \compr (\ihat) -\hat{ \jbar}_2$\label{line2:j2} 
\colorcomment{Choose the most likely estimate:}
\State $\jhat = \mathrm{argmin}_{\jhat \in \{\jhat_1, \jhat_2\}} (\Delta(x, \enc{\calg}(\jhat)))$ 
\State \Return $\jhat$
\end{algorithmic}
\end{algorithm}

\begin{algorithm}[H]
\caption{\compr: Compute $r_i$, given $i$.} 
    \begin{algorithmic}
        \State \textbf{Input:} $ i \in \sett{0,\ldots, 2^{k'k}-1}$
        \State $\hat{r}_i =i \cdot (n+1) \cdot B$ 
        \For{$z \in \sett{0,\ldots, k'k-1} $}
          \State $\hat{r}_i = \hat{r}_i +\lfloor  \frac{ i + 2^z }{2^{z+1} } \rfloor \left( \|a_z\|   + \frac{2L}{\log(kk')}  \cdot  \|\mathrm{bin}(z)\|\right) $ 
             \gComment{$a_z$ is the $z$'th row of $A$}
             
             \gComment{$\mathrm{bin}(z)$ is the binary expansion of $z$}

             \gComment{$\|\cdot\|$ denote Hamming weight}
        \EndFor
        \State \textbf{Return:} $\hat{r}_i$
    \end{algorithmic}
    \label{alg:compr}
\end{algorithm}

   \section{Analysis}\label{sec:analysis}
   In this section we analyze \Cref{alg:Dec}, proving a few statements that will be useful for our final proof of \Cref{thm:main} in \Cref{sec:last}.
   We start by setting up a bit more notation. Throughout this section, we assume that the codeword that was transmitted was  $g_j = \enc{\calg}(j) \in \calg,$ for some integer $j\in [r_i, r_{i+1})$ for some $i$. 

\subsection{Running time of \Cref{alg:Dec}}
We begin by analyzing the running time of \Cref{alg:Dec}.  In particular, we prove the following proposition.

   \begin{proposition} \label{prop:running-time}
        For a code $\mathcal{D}$ of length $D$, let $T_{\enc{\mathcal{D}}}(D)$ and $T_{\dec{\mathcal{D}}}(D)$ denote the running time of $\mathcal{D}$'s encoding map $\enc{\mathcal{D}}$ and $\mathcal{D}$'s decoding map $\dec{\mathcal{D}}$, respectively.
        Given the codes $\Cout, \Cin$ and our Gray code $\calg$ defined in \Cref{def:calg}, it holds that: 
        \begin{enumerate}
            \item $\enc{\calg}$ runs in time 
            \[
            O\left( d^2 \right) + O\left(T_{\enc{\Cout}}(n) + n\cdot T_{\enc{\Cin}}(n')\right)
            \]
            \item $\dec{\calg}$, which is given by Algorithm~\ref{alg:Dec}, runs in time \[
                O (n \cdot B) + O\left(n\cdot T_{\dec{\Cin}}(n')\right) + O\left( T_{\dec{\Cout}}(n)\right) + O(d)\;.
            \]
        \end{enumerate}
    \end{proposition}
    \begin{proof}
        We start with the encoding of $\calg$, which consists of the following steps.
        \begin{itemize}
            \item Given an integer $j$, we need to compute the $i$ for which $j\in [r_i, r_{i+1})$. 
            Recall that given $i$, \compr\ computes $r_i$. Thus, ind the corresponding $i$ by performing binary search on the domain $i\in \{0,\ldots, 2^{kk'} - 1\},$ calling $\compr$ in each iteration. Thus, the complexity of this step is $O(kk')$ times the time it takes to perform \compr. 

            We are left with analyzing the complexity of \compr. The loop inside it runs for $kk'$ iterations and in every iteration we perform a constant number of operations (multiplication, addition, and division) on $kk'$-bit integers. Note also that $||a_z||$ and $||\text{bin}(z)||$ can be computed in $O(kk')$.
            Now, as multiplication of two $kk'$-bit integers can be done in $O(kk'\log(kk'))$ time \cite{harvey2021integer}, the total running time of \compr\ is $O((kk')^2\cdot \log (kk'))$. 

            In total, the running time to find $i$ given $j$ is $O((kk')^3\cdot \log (kk')) \leq \tilde{O}(d^3)$.
            
            \item Given $i$ from the previous step, we encode $i$ to $c_i$ by first computing the message $\calr_{kk'}(i) \in \F_2 ^{kk'}$. This can be done by simply invoking the recursive definition given in \Cref{def:BRC} which runs in time $O(kk')=O(d)$.
            Then we encode the message $\calr_{kk'}(i)$ with $\enc{\Cout}$ and $\enc{\Cin}$ to a codeword $c_i\in F_2^{nn'}$. This can be performed in time $O(T_{\enc{\Cout}}(n) + n\cdot T_{\enc{\Cin}}(n'))$.  
            Thus, the final complexity of this step is $O(d) + O(T_{\enc{\Cout}}(n) + n\cdot T_{\enc{\Cin}}(n'))$.
            \item Given $c_i$ from the previous step, we next compute $w_i \in \calw$.  This involves computing $z_i$ and encoding it with the repetition code to obtain $L_{z_i} = \call(z_i)$; and computing the ``buffer'' sections $s_m$.  Adding the buffers and encoding $z_i$ clearly take time $O(d)$. Computing $z_i$ from $i$ can be done in time $O(d)$ as follows. We compute $\calr_{kk'}(i-1)$ in time $O(d)$, and since we already computed $\calr_{kk'}(i-1)$ in the previous step, we can identify $z_i$ by searching the only bit that differs between $\calr_{kk'}(i-1)$ and $\calr_{kk'}(i)$.
            \item At the end of the previous step, we have $w_i$. 
 We can repeat the process to obtain $w_{i+1}$.  Then we may obtain $g_j$ in time $O(d)$ from $w_{i+1}$ by flipping $j-r_i$ bits of $w_i$ (namely, the first $j-r_i$ bits on which $w_i$ and $w_{i+1}$ differ).  
        \end{itemize}
        Thus the overall running time of the encoder $\enc{\calg}$ is 
        \[
            \tilde{O}\left( d^3 \right) + O\left(T_{\enc{\Cout}}(n) + n\cdot T_{\enc{\Cin}}(n')\right)
        \]
        We proceed to analyze the running time of the decoder $\dec{\calg}$, given in \Cref{alg:Dec}.  We go line-by-line through \Cref{alg:Dec}.  
        \begin{enumerate}
            \item In Line~\ref{line4:maj-dec}, we take the majority of $B$ bits, for each $m \in [n+1]$.  This takes time $O(nB)$. 
            \item In lines~\ref{line4:lhat-start}-\ref{line4:lhat-comp}, we compute $\elhat$.  This takes time $O(n)$, as we apply $\dec{\calu}$ and $\dec{\calu^{\mathrm{comp}}}$ once each to a vector of length $n+1$; and then compute two Hamming distances between vectors of length $n+1$.  By Lemma~\ref{lem:unary-dec-complexity}, the former takes time $O(n)$, and the latter clearly also takes time $O(n)$.
            \item In Lines~\ref{line4:dec-inn-code-s}-\ref{line4:dec-inn-code-e}, Algorithm~\ref{alg:Dec} decodes $n$ inner codewords.  This takes time $O(n\cdot T_{\dec{\Cin}}(n'))$.
            \item In Lines~\ref{line4:getest}-\ref{line4:bgetest}, \Cref{alg:Dec} calls either \Cref{alg:est} or \Cref{alg:best}.  The running time of each of these includes:
            \begin{itemize}
                \item The time to update $\sighat$.  In \Cref{alg:est}, this includes time $O(L) = O(d)$ to decode the repetition code $\call$ to obtain $\hat{z}$; and then time $O(nn') = O(d)$ to perform the update.  In \Cref{alg:best}, the only work is setting $\sighat[m] = \bot$ for appropriate values of $m$, which runs in time $O(n) = O(d)$ as well.
                \item The time to decode $\sighat$ using $\dec{\Cout}$.  This takes $T_{\dec{\Cout}}(n)$ time.
                \item The time to decode the unary code $\calu$ (once for \Cref{alg:est}, twice for \Cref{alg:best}).  By \Cref{lem:unary-dec-complexity}, this takes time $O(d)$.
                \item The time to call $\compr$ (once for \Cref{alg:est}, twice for \Cref{alg:best}).  This takes time $\tilde{O}(d^2)$.
                \item 
                Finally, \Cref{alg:best} picks whichever of the two estimates $\jhat_1$ and $\jhat_2$ is better.  As written in \Cref{alg:best}, this requires computing $\enc{\calg}(\jhat_1)$ and $\enc{\calg}(\jhat_2)$, which naively would include an $O(d^2)$ term in its running time as above.  However, the only reason for the $O(d^2)$ term is the time needed to find $i$ given $j$.  In this case, we already have the relevant $i$ (it is the $\ihat$ returned by $\dec{\Cout}$), and so this step can be done in time $O(d) + O(T_{\enc{\Cout}}(n) + n \cdot T_{\enc{\Cin}}(n'))$ as well. 
                
                We note that several times throughout \Cref{alg:est} and \Cref{alg:best}, the algorithm needs access to $w_{\tilde{i}}$ for some value of $\tilde{i}$; these can be computed in the same way as $\enc{\calg}(\jhat_1)$ above, and so are covered by the $O(d) + O(T_{\enc{\Cout}}(n) + n \cdot T_{\enc{\Cin}}(n'))$ term.

            \end{itemize}
       
        \end{enumerate}
        Overall, the decoding complexity is
        \[
        O (n \cdot B) + O\left(n\cdot T_{\dec{\Cin}}(n')\right) + O\left( T_{\dec{\Cout}}(n)\right) + O\left(T_{\enc{\Cout}}(n) + n\cdot T_{\enc{\Cin}}(n')\right) + \tilde{O}(d^2)\;.
        \]
      
    \end{proof}

\subsection{Analyzing the failure probability of \Cref{alg:Dec}}

Our main result in this section says that the estimate $\jhat$ returned by \Cref{alg:Dec} is close to $j$ with high probability.

    \begin{theorem}\label{prop:decoding-alg-succeeds}
    Fix a constant $p \in (0,1/2)$.
        Let $q = 2^{k'}$ for a large enough integer $k'$. Let $\Cout$ be an $[n,k]_q$ linear code with relative distance $\delout$ that can decode efficiently from $e$ errors and $t$ erasures as long as $2e + t < \delout n$. 
        
        Let $\Cin$ be an $[n',k']_2$ linear code and suppose that $\pCinFail = o(1)$, where the asymptotic notation is as $n' \to \infty$.
        Let $\calg: \{0,\ldots,N -1\} \rightarrow \{0, 1\}^d$ be the Gray code defined in \Cref{def:calg} with $\Cout$ as an outer code and $\Cin$ as an inner code.  
        Suppose that the parameter $L$ in \Cref{def:calg} satisfies $L = \omega(\log(kk') \log\log(kk'))$, and suppose that the parameter $B$ in \Cref{def:calg} is an absolute constant (independent of $k,k', n,n', N$).
       Let $B, \beta, \xi > 0$ be constants  so that
        \begin{equation}\label{eq:decode-buffers-condition}
            2\exp(-C_p B)) < \beta < 1/4,
        \end{equation}
        where $C_p$ is a constant\footnote{The value of $C_p$ is determined in the proof; see \Cref{clm:maj-fail-prob}.} depending only on $p$; and
        \begin{equation}\label{eq:decode-outer-condition}
            2(1 + \xi)\pCinFail + 2\beta < \delout \;.
        \end{equation}

         Let $j\in \{0,\ldots,N -1\} $ and let $g_j = \enc{\calg}(j)$. Let $x\in \{0,1\}^d$ be the string $x = g_j \oplus \eta$ where $\eta \sim \ber(p)^d$ (the result of transmitting $g_j$ through the BSC$_p$). 
        Let $\jhat$ be the output of Algorithm~\ref{alg:Dec} when given as input the string $x$. Then for sufficiently large $t$ (relative to constants that depend on the constants $p,B,\beta, \xi$ above),
\begin{equation}\label{eq:probfail} \Pr_{\eta}[|j - \jhat| > t] \leq \exp(-\Omega(\repblow)) + \exp(-\Omega(n)) + \exp(-\Omega(t)).
\end{equation}
    \end{theorem}
    Above, we emphasize that the constants inside the $\Omega(\cdot)$ notation in \eqref{eq:probfail} may depend on the constants $p, \beta, B, \xi$.

The rest of this section is devoted to the proof of \Cref{prop:decoding-alg-succeeds}.  In each of the following sub-sections, we analyze a different step of \Cref{alg:Dec}, and show that it is successful with high probability.  \Cref{prop:decoding-alg-succeeds} will follow by a union bound over each of these steps; the formal proof of \Cref{prop:decoding-alg-succeeds} is at the end of the section.
 
    \subsubsection{Estimating the location of the crossover}

    The purpose of the following claims is to show that, except with probability $\exp(-\Omega(n))$, \Cref{alg:Dec} correctly identifies the interval in which the crossover point occurs.
    Recall the definition of $\ell$ from  \eqref{eq:ell-def}.
    Our goal is to show that with high probability, the value $\elhat$ computed in Line~\ref{line4:lhat-comp} of \Cref{alg:Dec} will be close to $\ell$. 
    
    We start by a simple application of the Chernoff bound (given in \Cref{lem:chernoff}) and show that the probability that the majority decoding of a single chunk $s_m$ in Line~\ref{line4:maj-dec} fails in $\exp(-\Omega(B))$ (assuming that $h_{i,\bar{j}}$ didn't fall in $s_m$). 
    \begin{claim} \label{clm:maj-fail-prob}
    Let $\eta \sim \ber(p)^{B}$.  Then there is some constant $C_p > 0$ so that
        \(
        \Pr_\eta\left[\maj \left(1^{B} \oplus \eta \right) \neq 1 \right] = \exp(-C_p \cdot B).
        \)
    \end{claim}  
    \begin{proof}
        The majority fails if at least half of the bits are changed to $0$. 
        The expected number of $0$s in $1^{B} + \eta$ is $p \cdot B$. Thus, by Chernoff bound (\Cref{lem:chernoff}), the probability that the majority fails can be upper bounded by 
        \[
            \exp \left( - \frac{p \cdot B \cdot (\frac{1}{2p} - 1)^{2} }{3} \right) =\exp(-C_p B)
        \]
        where $C_p = p \cdot \left( \frac{1}{2p} -1 \right)^2/3.$
        
    \end{proof}

    Our next focus is to show that $\elhat$ (computed in line \ref{line4:lhat-comp} of Algorithm ~\ref{alg:Dec}), is ``close'' to $\ell$.
    The next claim considers three possible scenarios depending on the location of $\elhat$. In the first scenario, $\elhat \in (\gapp n, n - \gapp n)$ is ``in the middle'' of the codeword. 
    In this case, we show that with high probability, $\ell \in [\elhs, \elhe]$ where $\elhs = \elhat - \gapp n$ and $\elhe = \elhat + \gapp n$. 
    The other two cases are that $\elhat \notin (\gapp n, n- \gapp n)$ is ``in the boundary'' of the codeword (with one case for the beginning and one for the end). Here, we show that with high probability $\ell \in [0, \elhe] \cup [\elhs, n]$ where $\elhe$ and $\elhs$ are defined according to lines \ref{line2:set-elhs-elhe-start}-\ref{line2:set-elhs-elhe-end} of Algorithm~\ref{alg:best} (See also Figure~\ref{fig:ell}). 
Formally, we have the following claim.

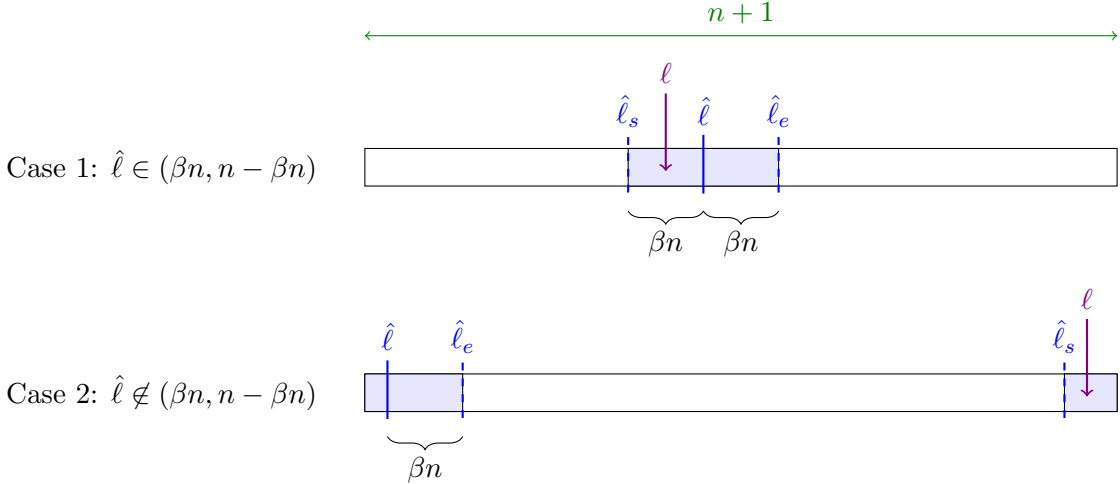
\begin{figure}
    \centering
\begin{tikzpicture}
\begin{scope}[yshift=1cm]
\draw [green!50!black, <->] (0,1) to (10,1); 
        \node[green!50!black] at (5,1.3) {$n+1$};
    \end{scope}
    \begin{scope}
    \node[anchor=east] at (-.5,.25) {Case 1: $\elhat \in (\gapp n, n- \gapp n)$};
     \draw (0,0) rectangle (10,.5);
     \draw[fill=blue!10] (3.5,0) rectangle (5.5,.5);
  \node[violet](ell) at (4, 1.5) {$\ell$};
  \draw[thick,violet,->] (ell) to (4,.2);
  \node[blue](ellhat) at (4.5, 1) {$\hat{\ell}$};
  \draw[thick, blue] (ellhat) to (4.5,-.1);
  \node[blue](elle) at (5.5, 1) {$\hat{\ell}_e$};
  \node[blue](ells) at (3.5, 1) {$\hat{\ell}_s$};
  \draw[thick, blue, dashed] (elle) to (5.5,-.1);
  \draw[thick, blue,dashed] (ells) to (3.5, -.1);
  \draw [decorate,decoration={brace,amplitude=5pt,mirror,raise=2ex}]
  (3.5,0) -- (4.5,0) node[midway,yshift=-2em]{$\gapp n$};
    \draw [decorate,decoration={brace,amplitude=5pt,mirror,raise=2ex}]
  (4.5,0) -- (5.5,0) node[midway,yshift=-2em]{$\gapp n$};
    \end{scope}

        \begin{scope}[yshift=-3cm]
    \node[anchor=east] at (-.5,.25) {Case 2: $\elhat \not\in (\gapp n, n- \gapp n)$};
     \draw (0,0) rectangle (10,.5);
     \draw[fill=blue!10] (9.3,0) rectangle (10,.5);
     \draw[fill=blue!10] (0,0) rectangle (1.3,.5);
  \node[violet](ell) at (9.6, 1.5) {$\ell$};
  \draw[thick,violet,->] (ell) to (9.6,.2);
  \node[blue](ellhat) at (.3, 1) {$\hat{\ell}$};
  \draw[thick, blue] (ellhat) to (.3,-.1);
  \node[blue](elle) at (1.3, 1) {$\elhe$};
  \node[blue](ells) at (9.3, 1) {$\elhs$};
  \draw[thick, blue, dashed] (elle) to (1.3,-.1);
  \draw[thick, blue,dashed] (ells) to (9.3, -.1);
  \draw [decorate,decoration={brace,amplitude=5pt,mirror,raise=2ex}]
  (.3,0) -- (1.3,0) node[midway,yshift=-2em]{$\gapp n$};
    \end{scope}
\end{tikzpicture}
    \caption{Two cases for where $\hat{\ell}$ can land.  
    As one case see in Case 2, it can be the case that $\ell$ is in the end of the transmitted codeword whereas $\elhat$, our estimate of $\ell$, is in the beginning.}
    \label{fig:ell}
\end{figure}

    \begin{claim} \label{clm:ell-estimation}
    Assume the conditions of \Cref{prop:decoding-alg-succeeds}.
        Let $\elhat$ be the value obtained in line \ref{line4:lhat-comp} of Algorithm ~\ref{alg:Dec}. Define the bad event $E_{\elhat}$ according to the following cases:
        \begin{enumerate}
            \item If $\elhat \in (\gapp n, n - \gapp n)$, then $E_{\elhat}$ is the event that $\ell \notin [\elhs, \elhe]$ where $\elhs = \elhat - \gapp n$ and $\elhe = \elhat + \gapp n$.
            \item If $\elhat \leq \gapp n$, then $E_{\elhat}$ is the event that $\ell \notin [0, \elhe] \cup [\elhs, n]$ where $\elhe = \elhat + \gapp n$ and $\elhs = n + 1 - (\elhat - \gapp n)$.
            \item If $\elhat \geq n - \gapp n$, then $E_{\elhat}$ is the event that $\ell \notin [0, \elhe] \cup [\elhs, n]$ where $\elhe = \elhat + \gapp n - (n + 1)$ and $\elhs = \elhat - \gapp n$.
        \end{enumerate}
        Then, the probability (over the choice of $\eta \sim \ber(p)^d$) that $E_{\elhat}$ occurs is at most $\exp(- \Omega_{B, \gapp, p}(n))$. 
    \end{claim}
    \begin{proof}
        We begin with Case 1, namely that $\elhat \in (\gapp n, n - \gapp n)$. If $\ell \notin [\elhs, \elhe]$, then $|\ell - \elhat| > \gapp n$. Assume without loss of generality that $\ell > \elhat$. 
        Let $\shat \in \zo^{n+1}$ be the quantity computed in \Cref{alg:Dec}, and suppose that $\elhat = \dec{\calu}(\shat)$.  (Note that $\elhat$ is either $\dec{\calu}(\shat)$ or $\dec{\calu^{\mathrm{comp}}}(\shat)$; assume without loss of generality that it is $\dec{\calu}(\shat)$, and the other case follows by an identical argument.)
        By the definition of the unary decoder, it must be that $\Delta (\shat, \enc{\calu}(\elhat)) \leq \Delta (\shat, \enc{\calu}(\ell))$. This implies that the number of zeros in $\shat[\elhat: \ell]$ is greater than the number of ones in this interval.
        This means that at least $\gapp n/2$  values in $\shat$ were decoded incorrectly by the majority decoder. 
        By \Cref{clm:maj-fail-prob}, the probability that a single value of $\shat$ was decoded incorrectly is $\exp(-C_p B)$. Thus, the expected number of values that are decoded incorrectly is $\exp(-C_p B) \cdot (n + 1)$.  By \Cref{lem:chernoff}, as long as $\gapp / 2 > \exp(-C_p B)$ (which it is by assumption), the probability that $\ell \notin [\elhs, \elhe]$ is at most $\exp(-\Omega(n))$, where the constant in the $\Omega(\cdot)$ depends on $B, p$ and $\beta$.

        Next, consider Case 2, namely that $\elhat \leq \gapp n$ and that $\ell \notin [0, \elhe] \cup [\elhs, n]$ where $\elhe = \elhat + \gapp n$ and $\elhs = n + 1 - (\elhat - \gapp n)$. Note that in this case it must be that $\ell > \elhat$ and that $|\ell - \elhat| > \gapp n$. Following the same arguments as in Case 1, we get again that the probability $E$ occurs is $\exp(-\Omega(n))$, and Case 3 follows in the same way.
    \end{proof}

    \begin{remark}[The meaning of ``close'']
        We remark that in Cases 2 and 3 \Cref{clm:ell-estimation} it can be the case that $\ell$ and $\elhat$ are not close to each other, in the sense that $|\ell - \elhat|$ is much bigger than $\gapp n$. This can happen if $\ell \leq \gapp n$, but $\elhat \geq n - \gapp n$ or vice versa, as depicted in \Cref{fig:ell}. However, if we consider the values modulo $n+1$, so that the interval $[0,n+1]$ ``wraps around,'' then  $\elhat$ and $\ell$ are actually close to one another. In this sense, \Cref{clm:ell-estimation} says that $\ell$ and $\elhat$ will be ``close'' to each other with high probability.
    \end{remark}

    \subsubsection{Decoding $z$}
    In this subsection, we are interested only in the case where $\elhat \in (\gapp n , n - \gapp n)$. This is because only \Cref{alg:est} (not \Cref{alg:best}), attempts to estimate $z$, and \Cref{alg:est} is only called when $\elhat$ is in the middle.\footnote{Intuitively, the reason that \Cref{alg:best} does not need to estimate $z$ is because when it is called, the cross-over point is near the boundary.  This means that the $g_j$ is already close to a codeword $w_i$, and we do not need to ``translate'' the symbols of $\sigma_i$ into $\sigma_{i+1}$ or vice versa.  Thus, estimating $z_{i+1}$ is not necessary.}  In this case, 
  \Cref{alg:est} decodes the first $\rowenc$ bits of $x$, which should contain the information $z_{i+1}$.  The following claim shows that this decoding process succeeds with probability $\exp(-\Omega (\repblow))$.

    \begin{claim} \label{clm:decode-rowenc}
        Assume the conditions of \Cref{prop:decoding-alg-succeeds}.
        Let $E_{\zhat}$ be the bad event that (a) $\elhat \in (\gapp n, n-\gapp n)$, and that (b) the quantity $\zhat$ computed on Line~\ref{line3:decode-alpha-z} of \Cref{alg:est} is incorrect, meaning that $\zhat \neq z_{i+1}$.
Suppose that the bad event $E_{\elhat}$ defined in \Cref{clm:ell-estimation} does not occur.  Then, conditioned on that, the probability that $E_{\zhat}$ occurs is at most 
\[ \Pr_{\eta}[E_{\zhat} \,|\, \overline{E_{\elhat}}] \leq \exp(-\Omega (\repblow)). \]
    \end{claim}
    \begin{proof}
        Clearly, if $\elhat \notin (\gapp n , n - \gapp n)$ the claim trivially holds as the probability of $E_{\zhat}$ is $0$.
        Assume that $\elhat\in (\gapp n , n - \gapp n)$ and that $E_{\elhat}$ did not occur.
        Recall that $\call(z_{i+1})$ simply duplicates each bit of $z_{i+1}$ $\repblow$ times and that $\tilde{L}$ is a noisy version of $\call(z_{i+1})$.  Let us write $\tilde{L} = \tilde{L}_0 \circ \tilde{L}_1 \circ \cdots \circ \tilde{L}_{\log(kk')-1}$, where each $\tilde{L}_m \in \zo^{\repblow}.$
        The decoding algorithm $\dec{\call}$ then takes a majority vote of each $\tilde{L}_m$ to recover the estimate $\zhat$.  Thus, it fails if and only if there is some $m \in [\log(kk') - 1]$ so that $\tilde{L}_m$ has more than half of its $\repblow$ bits flipped.
        By the Chernoff bound (\Cref{lem:chernoff}), and since the expectation of the number of bits that are flipped is $p \cdot \repblow$, the probability that this occurs for a particular $m$ most $\exp(-\Omega (\repblow) )$.
        Applying the union bound over all $\log(kk')$ values of $m$, the probability that the decoding of $\call(z_{i+1})$ fails is at most $\log(kk') \cdot \exp(-\Omega(\repblow)) = \exp(-\Omega(\repblow))$. 
    \end{proof}

    \subsubsection{Estimating $i$}
    In this subsection, we show that the estimate of $\ihat$ obtained in either \Cref{alg:est} or \Cref{alg:best} (depending on which was called by \Cref{alg:Dec}) succeeds with high probability.

   In more detail, $\ihat$ is computed either in \Cref{alg:est} or \Cref{alg:best} based on the location of $\elhat$ (whether it is in the middle or in the boundary).  The following claim shows that with high probability, the estimate of $\ihat$ is correct, meaning that $\ihat = i$ if $\elhat$ is in the middle, and $\ihat$ is equal to either $i$ or $i+1$ if $\elhat$ is in the boundary.
    
    \begin{claim} \label{clm:outer-code-decoding}
        Assume the conditions of \Cref{prop:decoding-alg-succeeds}. 
        Define the bad event $E_{\ihat}$ according to the following cases
        \begin{enumerate}
            \item If $\elhat \in (\gapp n, n - \gapp n)$ then $E_{\ihat}$ is the event that $\ihat \neq i$ after performing line \ref{line3:dec-out} in Algorithm~\ref{alg:est}.
            \item If $\elhat \notin (\gapp n, n - \gapp n)$ then $E_{\ihat}$ is the event that after performing line~\ref{line2:dec-out} in Algorithm~\ref{alg:best}, either
            \begin{itemize}
               \item $\elhat \leq \gapp n$ and $\ihat \neq i$; or
                \item $\elhat \geq n - \gapp n$ and $\ihat \neq i+1$.
            \end{itemize} 
        \end{enumerate}

        Assume that neither $E_{\elhat}$ nor $E_{\zhat}$ occurred. Then, conditioned on that, the probability that $E_{\ihat}$ occurs is at most 
\[ \Pr_{\eta}[E_{\ihat} \,|\, \overline{E_{\elhat}}, \overline{E_{\zhat}}] \leq \exp(-\Omega(n)), \]
where the constant in the $\Omega(\cdot)$ depends on $p, \xi$ and $\gapp$ defined in (\ref{eq:decode-outer-condition}).
      
    \end{claim}
    \begin{proof}
        Based on the location of $\elhat$, this claim 
        considers the lines \ref{line3:set-corrupt-outer-code}-\ref{line3:dec-out} in Algorithm~\ref{alg:est} and lines \ref{line2:set-corrupt-outer-code}-\ref{line2:dec-out} in Algorithm~\ref{alg:best}.
        We consider the two cases separately, and show that in each case, before $\ihat$ is computed, $\sighat$ corresponds to a noisy version of $\sigma_i$ or $\sigma_{i+1}$.  Then we can invoke the guarantees of $\dec{\Cout}$ to argue that we correctly return either $i$ or $i+1$.
        \begin{enumerate}
            \item \textbf{$\elhat$ is in the middle, i.e., $\elhat \in (\gapp n, n - \gapp n)$.}

            Since we assume that $E_{\elhat}$ did not occur, the quantity $\ell$ defined in \eqref{eq:ell-def} satisfies $\ell \in [\elhat_s, \elhat_e]$.  This implies that the chunks $\tilde{c}_m$ for $m\in [1, \elhs)$ are corrupted versions of the inner codewords of $c_{i+1}[m], m\in [1, \elhs]$; and that the chunks $\tilde{c}_m$ for $m\in (\elhe, n]$ are corrupted versions of the inner codewords of $c_{i}[m], m\in (\elhe, n]$. 

              First, we argue that with high probability, the chunks $\tilde{c}_m$ are \emph{mostly} correctly decoded to the symbols $\sighat[m]$ in Lines~\ref{line4:dec-inn-code-s}-\ref{line4:dec-inn-code-e} in \Cref{alg:Dec}, in the sense that with probability at least $\exp(-\Omega(n))$ over the choice of $\eta$, after Line~\ref{line4:dec-inn-code-e}, at least a $1 - (1 + \xi)\cdot \pCinFail$ fraction of  $m \in [1,\elhs) \cup (\elhe, n]$ satisfy
              \begin{equation}\label{eq:decode_good}
              \sighat[m] = \begin{cases} \sigma_i[m] & m \in (\elhe, n] \\ \sigma_{i+1}[m] & m \in [1, \elhs) \end{cases}
              \end{equation}
              where we recall that $\sigma_i \in \Cout$ is the $i$'th outer codeword, so that $c_i = \sigma_i \circ \Cin$.
              To see that \eqref{eq:decode_good} holds for most $m$, we observe that for any $m \in [1, \elhs) \cup (\elhe, n]$, the probability that $\sighat[m]$ does \emph{not} satisfy \eqref{eq:decode_good} is at most $\pCinFail,$ and so the expected number to not satisfy \eqref{eq:decode_good} is $\pCinFail (n - 2 \gapp n - 1)$.  Thus, by a Chernoff bound (Lemma~\ref{lem:chernoff}), the probability that more than $(1 + \xi) \pCinFail (n - 2\gapp n - 1)$ of these $m$ do not satisfy \eqref{eq:decode_good} is at most $\exp(-\Omega(n))$, where the constant inside the $\Omega(\cdot)$ depends on $p, \gapp$ and $\xi$.

            Next, we argue that if the favorable case above occurs, then $\sighat$ is a noisy version of $\sigma_i$, with not too many errors or erasures.
            
            Since we assume that $E_{\zhat}$ did not occur, we have that $\zhat = z_{i+1}$.
            Thus, our choice of ordering of the codewords of $\cC$ (see Equation~\eqref{eq:sig-row}) implies that
            \[ c_{i} = c_{i+1} \oplus a_{z_{i+1}} = c_{i+1} \oplus a_{\zhat}.\] Therefore, after $\sighat$ is done being updated in \Cref{alg:est} (that is, after line~\ref{line:doneupdating}), for the $m \in [1, \elhs) \cup (\elhe, n]$ that satisfy \eqref{eq:calg-rate}, we have
            \[
                \sighat[m] = 
                \begin{cases}
                    c_{i+1}[m] \oplus a_{\zhat}[m], \text{ if $m \in [1, \elhs)$} \\
                    c_i[m], \text{ if $m \in (\elhe,n]$}
                \end{cases} 
           =
                    c_i[m]
            \]
            Meanwhile, we have $\sighat[m] = \bot$ for all $m \in [\elhs, \elhe].$  In the favorable case that \eqref{eq:decode_good} is satisfied for all but $(1 + \xi) \pCinFail (n - 2\gapp n - 1)$ values of $m \in [1, \elhs) \cup (\elhe, n]$, we conclude that $\sighat \in \F_q^n$ is a noisy version of $\sigma_i \in \Cout$, where there are at most $2\gapp n + 1$ erasures, and at most $(1 + \xi) \pCinFail \cdot (n - 2\gapp n - 1)$ errors.

            \item \textbf{$\elhat$ is in the boundary, i.e., $\elhat \notin (\gapp n, n - \gapp n)$.}

Again, since we assume that $E_{\elhat}$ did not occur, the quantity $\ell$ defined in \eqref{eq:ell-def} satisfies $\ell \in [1, \elhe] \cup [\elhs, n]$.  If $\ell \in [1, \elhe]$, then for all $m \in (\elhe, \elhs)$, the chunk $\tilde{c}_m$ is a corrupted version of the inner codeword $c_{i}[m]$.  If $\ell \in [\elhs, n]$, then for all $m \in (\elhe, \elhs)$, the chunk $\tilde{c}_m$  is a corrupted version of the inner codeword $c_{i+1}[m]$.
By the same argument as in the previous case, we conclude that with probability at least $1 - \exp(-\Omega(n))$ over the choice of $\eta$, after Line~\ref{line4:dec-inn-code-e} in \Cref{alg:Dec}, at least a $1 - (1 + \xi) \cdot \pCinFail$ fraction of the $m \in (\elhe, \elhs)$ satisfy
\begin{equation*}
\sighat[m] = \begin{cases} \sigma_i[m] & \ell \in [1, \elhe] \\ \sigma_{i+1}[m] & \ell \in [\elhs, n] \end{cases}
\end{equation*}
Since we have $\sighat[m] = \bot$ for all $m \in [1, \elhe] \cup [\elhs,n]$,
this means that with probability $1 - \exp(-\Omega(n))$, when $\ihat$ is computed on Line~\ref{line2:dec-out}, $\sighat$ is a corrupted version of either $\sigma_i$ or $\sigma_{i+1}$, with at most $2\gapp n + 1$ erasures and at most $(1 + \xi)\pCinFail(n - 2\gapp n - 1)$ errors.     
        \end{enumerate}
        
Thus, in either case, we have that when $\dec{\Cout}$ is called (Line~\ref{line3:dec-out} for \Cref{alg:est} or Line~\ref{line2:dec-out} for \Cref{alg:best}), it is called on a corrupted codeword $\sigma$ that has at most $2 \gapp n + 1$ erasures and at most $(1 + \xi)\pCinFail \cdot(n - 2\gapp n - 1)$ errors.
     Recall that our outer code can recover efficiently from $e$ errors and $t$ erasures, as long as $2e + t < \delout n$.  Plugging in the number of errors and erasure above, we see that indeed we have 
        \begin{align*}
            2e + t &\leq 2(1 + \xi)\pCinFail (n - 2\gapp n -1) + 2\gapp n + 1 \\
            &\leq \left( 2(1 + \xi)\pCinFail + 2\beta + \frac{1}{n} \right) \cdot n \\
            & < \delout n
        \end{align*}
        where the last inequality holds for large enough $n$ (relative to $\frac{1}{\delout}$) and by our inequality assumption \eqref{eq:decode-outer-condition}.
        We conclude that with probability at least $1 - \exp(\Omega(n))$:
        \begin{itemize}
            \item If $\elhat$ is in the middle, then $\ihat = i$
            \item If $\elhat$ is in the boundary, then $\ihat$ is either $i$ or $i+1$, depending on which side of the boundary $\elhat$ is on.
        \end{itemize}
        This proves the claim.
    \end{proof}
    \subsubsection{Estimating $j$}
Next, we argue that the estimate $\jhat$ that \Cref{alg:Dec} returns satisfies $|j -\jhat| = \Delta(g_j, g_{\jhat})$ with high probability.  Before we state and prove that (in \Cref{clm:j-jhat} below), we first prove the correctness of \Cref{alg:compr}, as this is used as a step in the process of deterimining $\jhat$.
 \begin{claim}\label{cl:compr_correct}
 \Cref{alg:compr} is correct.  That is, given $i \in \{0,1,\ldots, q^k-1\}$, $\compr(i) = r_i$.
 \end{claim}
  \begin{proof}
        Consider the task of computing $r_i$ from $i$.  From Definition~\ref{def:notation}, we have
        \begin{equation}\label{eq:sum}
     r_i = \sum_{t=1}^i \Delta(w_{t-1}, w_t).
     \end{equation}
        Recalling that we may break up the codewords $w_t \in \mathcal{W}$ into chunks, we see that there are three types of contributions to $r_i$:  (1) Contributions from the chunks $s_m$ for $m=1, \ldots, n+1$; (2) contributions from the chunks $\tilde{c}_m$ for $m=1, \ldots, n$; and (3) contributions from the chunks $\tilde{L}$.
        We consider each of these in turn.
\begin{enumerate}
        \item \textbf{The chunks $s_m$.}  Since $t-1$ and $t$ have different parities, the chunks $s_m$ in $w_{t-1}$ are all completely different from those in $w_t$.  This contributes a total of $(n+1) \cdot B \cdot i$ to the sum in~\eqref{eq:sum}.

        \item \textbf{The chunks $\tilde{c}_m$.}
        Recall from \eqref{eq:sig-row} that $c_t = c_{t-1} + a_{z_t}$, where $a_{z_t}$ is the $z_t$th row of the generator matrix $A$ of $\cC$.  Thus, the contribution to \eqref{eq:sum} of the chunks $\tilde{c}_m$ for $m=1, \ldots, n$ is
        \[ \sum_{t=1}^i \Delta(c_{t-1}, c_t) = \sum_{t=1}^i \| a_{z_t} \|,\]
        where $\|\cdot\|$ denotes hamming weight. 
        For each $z$, from Observation~\ref{obs:BRC}, the number of $t$ so that $z = z_t$ is $\left\lfloor \frac{i+2^z}{2^{z+1}} \right\rfloor.$  
        Thus, the total contribution to $r_i$ from the $\tilde{c}_m$ chunks is
        \[ \sum_{z=0}^{k'k - 1}  \lfloor \frac{i + 2^z}{2^{z+1}} \rfloor \cdot \|a_z\|.\]
        
        \item \textbf{The chunk $\tilde{L}$.} Recall that what goes into the chunk $\tilde{L}$ in $c_t$ is $\call(z_t)$, where $\call$ is the code that repeats each bit representing $z_t$ exactly $\repblow$ times.  Thus, the contribution to \eqref{eq:sum} from these chunks is
        \[ \sum_{t=1}^i \Delta(\call(z_{t-1}), \call(z_t)) = \repblow \sum_{t=1}^i \Delta( \mathrm{bin}(z_{t-1}), \mathrm{bin}(z_t)),\]
        where $\mathrm{bin}(z)$ denotes the binary expansion of $z$.  By Observation~\ref{obs:BRC}, $z_t > 0$ if and only if $t$ is even.  Thus, for all $z \in \{1, 2, \ldots, kk' - 1\}$, $z = z_t$ implies that $z_{t-1} = z_{t+1} = 0$.  This means that the contributions from the two terms
        \[ \Delta( \mathrm{bin}(z_{t-1}, z_t)) + \Delta( \mathrm{bin}(z_t, z_{t+1}))\]
        is given by $2\|\mathrm{bin}(z_t)\|$.  Again by Observation~\ref{obs:BRC}, the number of times each such $z$ appears as $z_t$ for some $t \leq i$ is $\left\lfloor \frac{i + 2^z }{2^{z+1}} \right\rfloor$.  Thus, the total contribution to \eqref{eq:sum} of these terms is
        \[ \repblow \sum_{z=1}^{kk'-1} \left\lfloor \frac{i + 2^z }{2^{z+1}} \right\rfloor \cdot 2\|\mathrm{bin}(z)\| = \repblow \sum_{z=0}^{kk'-1} \left\lfloor \frac{i + 2^z }{2^{z+1}} \right\rfloor \cdot 2\|\mathrm{bin}(z)\|\]
        where in the equality we have added back in the $t=0$ term as $\|\mathrm{bin}(0)\| = 0$ and this does not affect the sum.
                \end{enumerate}
            Finally, we observe that \Cref{alg:compr} exactly computes the three contributions above.  First, it initializes $\hat{r}_i$ to $(n+1)Bi$ to account for the $s_m$ chunks; and then it loops over all $z \in \{0,1,\ldots, kk'-1\}$ and adds the contributions from the $\tilde{c}_m$ and $\tilde{L}$ chunks.
    \end{proof}

    \begin{claim}\label{clm:j-jhat}
        Let $j\in [N] $ and set $i$ to  be such that, $j\in [r_i,r_{i+1})$. Further, let $x = g_j \oplus \eta$ be the noisy version of $g_j$ and $\jhat$ be the output of Algorithm~\ref{alg:Dec}.
        Let $t > 0$.
        If the bad events $E_{\elhat}, E_{\zhat}$, and $E_{\ihat}$ do not occur, then the probability that $|j - \jhat| > t$, conditional on this, is bounded by
        \[ 
        \Pr_\eta \left[ |j - \jhat| > t \,|\, \overline{E_{\elhat}}, \overline{E_{\zhat}}, \overline{E_{\ihat}} \right] \leq \exp(-\Omega(t))\;, 
        \] 
        where the constant inside the $\Omega(\cdot)$ depends on $p$.
    \end{claim}

\begin{proof}
As in the proofs of earlier claims, we separate the analysis into two scenarios: one where $\elhat$ is in the middle, and one where it is in the boundary.

  \begin{enumerate}
      \item \textbf{$\elhat \in (\gapp n, (1-\gapp)n)$ is in the middle.}
          
    In this case the function $\getest$ (\Cref{alg:est}) is invoked to compute $\jhat$. 
    Since we assume that $E_{\elhat}, E_{\zhat}$, and $E_{\ihat}$ all hold, we have $\ihat = i$.  By \Cref{cl:compr_correct}, given $i$, the value $\compr(\ihat)$ computed on Line~\ref{line3:jhat-est} is equal to $r_{\ihat}$ and hence equal to $r_i$. 
    
    Next, we will show that the estimate $\jhat$ computed in \Cref{alg:est} satisfies $\jhat \in [r_i, r_{i+1})$.  To see this observe that \Cref{alg:est} first computes $\hat{\bar{j}} = \dec{\calu}(x[H] \oplus w_{\ihat}[H]) = \dec{\calu}(x[H] \oplus w_i[H]).$ Recall from Observation~\ref{obs:unary} that since $j \in [r_i, r_{i+1})$, we have $g_j[H] \oplus w_i[H] = \enc{\calu}(\bar{j})$. (Here we are using the fact that the set $H = \sett{m\,:\, w_i[m] \neq w_{i+1}[m] }$ in \Cref{alg:est} is the set of elements that appear in the vector $h_i$).  Since $x = g_j \oplus \eta$, we see that
    \[ \hat{\bar{j}} = \dec{\calu}(\enc{\calu}(\bar{j}) \oplus \eta[H] ).\]
    Now, we consider the probability that $\hat{\bar{j}}$ is very different than $\bar{j}$. 
For any fixed $\hat{\bar{j}}$, we claim that
    \[ 
    \Pr_\eta\left[ \hat{\jbar} = \dec{\calu}(\enc{\calu}(\bar{j}) \oplus \eta[H]) \right] \leq \exp(-\Omega(|\hat{\jbar} - \bar{j}|)).
    \]
    Indeed, $\enc{\calu}(\bar{j})$ and $\enc{\calu}(\hat{\jbar})$ differ on $|\bar{j} - \hat{\jbar}|$ coordinates, and $\dec{\calu}$---which just finds the $\hat{\jbar}$ that is closest to the received word---will return $\hat{\jbar}$ rather than the correct answer $\bar{j}$ only if at least half of these bits are flipped by $\eta[H]$.  The probability that this occurs is 
    the probability that at least half of $|\hat{\jbar} - \bar{j}|$ i.i.d. random bits, distributed as $\ber(p)$, are equal to one.  As $p < 1/2$, by a Chernoff bound (Lemma~\ref{lem:chernoff}), the probablity that this occurs is at most $\exp(-\Omega(|\bar{j} - \hat{\jbar}|))$, where the constant inside the $\Omega(\cdot)$ depends on $p$.  Thus, the probability that $\dec{\calu}$ returns any $\hat{\jbar}$ with $|\hat{\jbar} - \bar{j}| \geq t$ is at most
    \begin{align}
    &\Pr_\eta\left[|\bar{j} - \hat{\jbar}| \geq t\right]\notag \\
    &\qquad = \Pr_{\eta}\left[| \bar{j} - \dec{\calu}(\enc{\calu}(\jbar) \oplus \eta[H] ) | \geq t \right] \notag \\
    &\qquad \leq \sum_{\hat{\jbar} \geq \jbar + t} \Pr\left[\hat{\jbar} = \dec{\calu}(\enc{\calu}(\jbar) \oplus \eta[H] ) \right]  + \sum_{\hat{\jbar} \leq \jbar - t} \Pr\left[ \hat{\jbar} = \dec{\calu}(\enc{\calu}(\jbar) \oplus \eta[H] ) \right]\notag\\
    &\qquad \leq 2\sum_{s \geq t} \exp(-\Omega(s))\notag \\
    &\qquad \leq \exp(-\Omega(t)). \label{eq:conc}
    \end{align}

    This shows that $\jbar$ is likely close to $\hat{\jbar}$.  As we observed above, the value $\compr(\ihat)$ computed by \Cref{alg:est} is equal to $r_i$, so we have
    \[ \jhat = \compr(\ihat) + \hat{\jbar} = r_i + \hat{\jbar},\]
    and by definition we have that
    \[ j = r_i + \jbar.\]
    Thus, $|j - \jhat| = |\bar{j} - \hat{\jbar}|$, and \eqref{eq:conc} implies that
    \[ \Pr\left[ |j - \jhat| \geq t \right] \leq \exp(-\Omega(t))\;,\]
    as desired.

\item \textbf{$\elhat$ is in the boundary.} 

In this case, the function  \bgetest\ is invoked.  
Since we assume that $E_{\elhat}, E_{\zhat}$, and $E_{\ihat},$  we have that 
$\ihat = i$ if $\ell \in [0,\elhe]$ and $\ihat = i+1$ if $\ell \in [\elhs, n]$.

First assume that $\ell \in [0,\elhe]$, and note that this implies  both that $\ihat = i$ and that the crossover point satisfies $h_{i,\jbar} \leq \rowenc + 2\gapp n(n'+B)$.  Further, by \Cref{cl:compr_correct}, we have $\compr(\hat{i}) = r_i$.  

Let
\[ H_1 = \{ m\ |\ w_{i+1}[m] \neq w_i[m], 0 \leq m \leq L + 2\gapp n(n' + B)\},\]
and let
\[ H_2 = \{m \ |\ w_{i}[m]  \neq w_{i -1}[m] ,   d -  2\gapp n ( n' + B ) \leq m \leq d\}.\]
(These are the two values of $H$ chosen in \Cref{alg:best} when computing $\jhat_1$ and $\jhat_2$, respectively; we give them separate names $H_1$ and $H_2$ for the analysis.) 

First we analyze the choice of $\jhat_1$.
By \Cref{obs:unary} and the fact that $h_{i,\bar{j}}\in H_1$, we have $g_j[H_1] \oplus w_i[H_1] = \enc{\calu}(\bar{j}),$ and so as above we have
\[ \hat{\bar{j}}_1 = \dec{\calu}( \enc{\calu}(\bar{j}) \oplus \eta[H_1]).\]
The same reasoning as in Case 1 implies that
\begin{equation}\label{eq:j1}
\Pr[ |{\hat{j}}_1 - j| \geq t/2 ] \leq \exp(-\Omega_p(t)).
\end{equation}
Further, in this case we also have that
\[ |\jhat_1 - j| = \Delta(g_j, g_{{\jhat}_1}).\]
Indeed, this follows because, regardless of the noise $\eta$, we have $\hat{\bar{j}}_1 \leq |H_1| \leq \Delta(w_i, w_{i+1})$, which means that 
    \[ \jhat = \compr(\hat{i}) + \hat{\bar{j}}_1 = r_i + \hat{\bar{j}}_1 \in [r_i, r_{i+1}).\]
    Now, since $j$ and $\hat{j}_1$ are both in the same interval $[r_i, r_{i+1})$, we must have $\Delta(g_j, g_{\hat{j}_1}) = |j - \hat{j}_1|$.  This is true because---assuming without loss of generality that $\jhat_1 \geq j$---to get from $g_j$ to $g_{\hat{j}_1}$, we flip the bits indexed by $h_i[j+1], h_i[j+2], \ldots, h_i[\jhat_1]$, and there are $|j - \hat{j}_1|$ such bits.

Next, we analyze $\jhat_2$. First, note that as with $\jhat_1$, we have
\begin{equation}\label{eq:same2}
|\jhat_2 - j | = \Delta(g_j, g_{\jhat_2}). 
\end{equation}
Indeed, by construction we have $\hat{\jbar}_2 \leq |H_2|$, which means that $\jhat_2 = r_i - \hat{\jbar}_2$ is towards the end of the interval $[r_{i-1}, r_i)$; concretely, it implies that the crossover point corresponding to $\jhat_2$ satisfies
\[ h_{i-1, \jhat_2 - r_{i-1}} \in [d - 2\gapp n(n' + B), d].\]
Thus, to get from the codeword $g_{\jhat_2}$ to the codeword $g_j$, we need to flip all of the bits indexed by $m$ in the set
\[ \{ m \, |\, w_i[m] \neq w_{i-1}[m], h_{i-1, \jhat_2 - r_{i-1}} \leq m \leq d \},\]
as well as all the bits indexed by $m$ in the set
\[ \{ m\,|\, w_i[m] \neq w_{i+1}[m], 0 \leq m \leq h_{i, \bar{j}} \}.\]
The number of elements in the first set is $\jhat_2 = r_i - \hat{\jbar}_2$, and the number in the second set is $\bar{j} = j-r_i$.  Since $\beta < 1/4$ and $h_{i, \bar{j}} \leq L + 2\beta n (n' + B)$ and $h_{i-1, \jhat_2 - r_{i-1}} \geq d - 2\beta n (n' + B)$, these two sets are disjoint.  
Thus the total number of indices we need to flip to get from $g_{\jhat_2}$ to $g_j$ is the sum of the sizes of these two sets, which is
\[ (r_i - \hat{\jbar}_2) + (j-r_i) = j - \hat{\jbar}_2,\]
establishing \eqref{eq:same2}.  A similar argument shows that $\Delta(g_{\jhat_1}, g_{\jhat_2}) = |\jhat_1 - \jhat_2|.$

Next, note that \Cref{alg:best} sets $\jhat = \jhat_2$ only if $\Delta(x, g_{{\jhat}_2}) \leq \Delta(x, g_{{\jhat}_1})$.  To analyze the probability that this occurs, fix $\jhat_1$ and $\jhat_2$, and define
\[A_1 :=  A_1(\jhat_1) := \{ m \in H_1\,|\, h_{i,\jbar} \leq m \leq h_{i,\jhat_1 - r_i} \}\]
and let 
\[ A_2 := A_2(\jhat_1, \jhat_2) := \{m \in H_2\,|\, m \geq h_{i-1, \jhat_2 - r_{i-1}} \} \cup \{ m \in H_1 \,|\, m \leq \min(h_{i,\bar{j}}, h_{i, \jhat_1 - r_i}) \}. \]
That is, $A_1$ is the set of indices $m$ so that $g_{\jhat_1}[m] \neq g_j[m]$ but $g_{\jhat_2}[m] = g_j[m]$, and $A_2$ is the set of indices $m$ so that $g_{\jhat_2}[m] \neq g_j[m]$ but $g_{\jhat_1}[m] = g_j[m]$.  Notice that $|A_1| + |A_2| = \Delta(g_{\jhat_1}, g_{\jhat_2})$.  Notice also that 
\begin{equation}\label{eq:A1A2} |A_1| = \begin{cases} \Delta(g_j, g_{\jhat_1}) = \jhat_1 - j & j \leq \jhat_1 \\ 0 & j > \jhat_1 \end{cases} \qquad \text{and} \qquad |A_2| =  \begin{cases} \Delta(g_j, g_{\jhat_2}) = j - \jhat_2 & j \leq \jhat_1 \\ \Delta(g_{\jhat_1}, g_{\jhat_2}) = \jhat_1 - \jhat_2 & j > \jhat_1 \end{cases}
\end{equation}

Now, consider the event that \Cref{alg:best} sets $\jhat = \jhat_2$, which happens only if $\Delta(x, g_{\jhat_2}) \leq \Delta(x, g_{\jhat_1})$.  
Note that 
\[
\Delta (x, g_{\jhat_2}) = \Delta (g_j \oplus \eta, g_{\jhat_2}) = |A_2| - \|\eta[A_2]\| + \|\eta[\bar{A_2}]\|
\]
and a similar expression holds for $\Delta (x, g_{\jhat_1})$. Therefore,
the event that $\Delta(x, g_{\jhat_2}) \leq \Delta(x, g_{\jhat_1})$ is the same as the event that
\[ |A_2| - \|\eta[A_2]\| + \|\eta[A_1]\| \leq |A_1| - \|\eta[A_1]\| + \|\eta[A_2]\|,\]
which, rearranging, is the same as the event that
\begin{equation}\label{eq:badevent}
|A_2| - 2\|\eta[A_2]\| - (|A_1| - 2\|\eta[A_1]\|) \leq 0.
\end{equation}

Note that $|A_2| - 2\|\eta[A_2]\|$ is a sum of $|A_2|$ independent random variables that are $+1$ with probability $1-p$ and $-1$ with probability $p$, and similarly for $|A_1| - 2\|\eta[A_1]\|$.  Moreover, 
since $A_1$ and $A_2$ are disjoint, the whole left hand side of \eqref{eq:badevent} is the sum of $|A_1| + |A_2|$ independent $\pm 1$-valued random variables, and the expectation of the left-hand-side is $(|A_2|-|A_1|)(1 - 2p)$, which is larger than zero when $|A_2| > |A_1|$ and $p < 1/2$.  By Hoeffding's inequality (Lemma~\ref{lem:hoeffding}), provided that $|A_2| > |A_1|$ and $p < 1/2$, the probability that \eqref{eq:badevent} occurs is at most
\[ \Pr_\eta[ \Delta(x,g_{\jhat_2}) \leq \Delta(x, g_{\jhat_1}) ] \leq 2 \exp\left( -\Omega_p\left( \frac{(|A_2|-|A_1|)^2 }{|A_1| + |A_2|}\right)\right),\]
where the constant inside the $\Omega_p(\cdot)$ depends on the gap between $p$ and $1/2$.

Now, consider the event $E$ that \Cref{alg:best} picks $\jhat_1$ and $\jhat_2$ so that all of the following hold:
\begin{itemize}
    \item[i.] $\Delta(x, g_{\jhat_2}) \leq \Delta(x, g_{\jhat_1})$
    \item[ii.] $|\jhat_1 - j| \leq t/2$
    \item[iii.] $|\jhat_2 - j| \geq t$
\end{itemize}
By a union bound, the probability that $E$ occurs is at most
\[ \Pr[E] \leq \sum_{\jhat_1, \jhat_2} \Pr[ \Delta(x, g_{\jhat_2}) \leq \Delta(x, g_{\jhat_1})]
\leq \sum_{\jhat_1, \jhat_2} 2\exp\left( -\Omega_p\left( \frac{(|A_2| - |A_1|)^2}{|A_1| + |A_2|} \right)\right)
,\]
where the sum is over all $\jhat_1$ and $\jhat_2$ that satisfy (ii) and (iii) above.  For any such $\jhat_1$ and $\jhat_2$, we have
\[ 
\frac{(|A_2| - |A_1|)^2}{|A_2| + |A_1|} \geq \frac{|\jhat_2 - j|}{6}.
\]
Indeed, (ii) and (iii) imply that $|\jhat_2 - j| \geq 2|\jhat_1 - j|$, so by \eqref{eq:A1A2}, if $j \geq \jhat_1$, then 
\[ \frac{(|A_2| - |A_1|)^2}{|A_2| + |A_1|} = |A_2| = |\jhat_1 - \jhat_2| \geq |\jhat_2 - j| - |\jhat_1 - j| \geq \frac{|\jhat_2 - j|}{2}, \] 
and if $j < \jhat_1$, then
\[ \frac{(|A_2| - |A_1|)^2}{|A_2| + |A_1|} = 
\frac{ ( (j - \jhat_2) - (\jhat_1 - j) )^2}{\jhat_1 - \jhat_2} \geq \frac{(|\jhat_2 - j|/2)^2}{(3/2)|j - \jhat_2|} = \frac{|\jhat_2 - j|}{6}.\]
Thus, we have 
\begin{align*}
\Pr[E] &\leq \sum_{\jhat_1, \jhat_2} \Pr[ \Delta(x, g_{\jhat_2}) \leq \Delta(x, g_{\jhat_1})] \\ &\leq 
\sum_{\jhat_1,\jhat_2} 2\exp\left(-\Omega_p(|\jhat_2 - j|)\right) \\
&\leq 2t \sum_{\jhat_2} \exp(-\Omega_p(|\jhat_2 - j|))\\
&\leq 2t \sum_{s \geq t} \exp(-\Omega_p(s)) \\
&\leq 2t\exp(-\Omega_p(t)) = \exp(-\Omega_p(t)).
\end{align*}
Above, we have used the fact from (ii) that there are at most $t$ values of $\jhat_1$ in the sum; then we have used (iii) (and the fact that in this case, \Cref{alg:best} will only choose $\jhat_2 < j$) to re-write the sum over $\jhat_2$ as a sum over $s \geq t$.

Altogether, by a union bound over the event $E$ analyzed above and the event that $|\jhat_1 - j| \geq t/2$ in \eqref{eq:j1}, we conclude that for all $t$ large enough (relative to the gap between $p$ and $1/2$), with probability at least $1 - 2\exp(-\Omega_p(t)) = 1 - \exp(-\Omega_p(t))$, we have both $|\jhat_1 - j| \leq t/2$ and also that $E$ does not occur.  Suppose that this favorable case happens.

Now consider $\jhat$.  If $\jhat = \jhat_1$, then by above, $|\jhat_1 - j| \leq t/2$ and so $|\jhat - j| \leq t/2$.  On the other hand, if $\jhat = \jhat_2$, then either $|\jhat - j| = |\jhat_2 - j| < t$, or else event $E$ occurs.  (Indeed, if $|\jhat_2 - j| \geq j$, then (iii) holds; (i) holds because \Cref{alg:best} chose $\jhat = \jhat_2$; and (ii) holds because we are assuming that the favorable case in \eqref{eq:j1} occurs).  But since we are assuming that $E$ does not occur, this implies that $|\jhat - j| \leq t$.  Either way, we conclude $|\jhat - j| \leq t$ except with probability $\exp(-\Omega_p(t))$, which proves the claim when $\elhat$ is on the boundary and $\ell \in [0, \elhat_e]$.

The above handled only the sub-case when $\ell \in [0,\elhat_e]$.  There is also the sub-case where $\ell \in [\elhat_s, n]$.  However, that case follows by an identical argument. 

  \end{enumerate}
  Thus, we have handled both the case when $\elhat$ is in the middle and the case where $\elhat$ is on the boundary, and this proves the claim.
\end{proof}

Finally, we are ready to prove \Cref{prop:decoding-alg-succeeds}.
    \begin{proof}[Proof of \Cref{prop:decoding-alg-succeeds}]
    \Cref{clm:j-jhat} implies that $|\jhat - j| \leq t$ with probability at least $1 - \exp(\Omega(t))$, provided that none of $E_{\elhat}$, $E_{\zhat}$, and $E_{\ihat}$  occur.  By Claims~\ref{clm:ell-estimation}, \ref{clm:decode-rowenc} and \ref{clm:outer-code-decoding} together with a union bound, the probability that any of these occur is at most $\exp(-\Omega(n)) + \exp(-\Omega(\repblow)) + \exp(-\Omega(n))$.  Thus, we conclude that with probability at least 
    \[ 1 - \exp(-\Omega(n)) - \exp(-\Omega(\repblow) - \exp(-\Omega(t)),\]
    we have $|\jhat - j| \leq t$, as desired.
    \end{proof}

    \section{Putting all together and choosing parameters}\label{sec:last}

In order to prove \Cref{thm:main}, we will plug in a Reed--Solomon (RS) code as $\Cout$.  Thus, before we prove the theorem, we recall the definition of RS codes and their basic properties.

\begin{definition}
    	Let $\alpha_1, \alpha_2, \ldots, \alpha_n$ be distinct points of the finite field $\mathbb{F}_q$ of order  $q$. For $k<n$ the $[n,k]_q$ \emph{RS code} 
		defined  by the  evaluation set $\{ \alpha_1, \ldots, \alpha_n \}$ is the   set of codewords 
		\[
		\left \lbrace \left( f(\alpha_1), \ldots, f(\alpha_n) \right) \mid f\in \mathbb{F}_q[x],\deg f < k \right \rbrace \;.
		\]
\end{definition}
It is well-known that RS codes are Maximum Distance Separable (MDS), which means in particular that an RS code of rate $\calr$ and distance $\delta$ has
\[ \calr = 1 - \delta + 1/n.\]
Moreover, encoding and decoding of RS codes can be done in $O(n \cdot \text{poly}(\log n))$ time (see e.g., \cite{gao2003new,lin2014novel}).

We are now ready to prove our main result, which we restate here for the reader's convenience. 
\mainResult*
    \begin{proof}[Proof of \Cref{thm:main}]
        Given \Cref{prop:decoding-alg-succeeds} and \Cref{prop:running-time} we are left to show that we can choose the outer code and the parameters $L, B, \beta, \xi$ such that: (i) the rate of our Gray code is $\Rin - \varepsilon$ where $\Rin$ is the rate of the inner code, (ii) that inequalities \eqref{eq:decode-buffers-condition} and \eqref{eq:decode-outer-condition} hold, (iii) that the running time of our encoder and decoder are as desired.
        
        Let $q = 2^{k'}$ where $k'$ is a large enough integer. Let $\Cin$ be as in the theorem statement. 
        We start with choosing the outer code. We shall use as $\Cout$ an $[n,k]_q$ Reed--Solomon code where the evaluation points are taken to be $\Fq^*$, namely, $n = q-1$. 
        As Reed--Solomon codes are MDS codes, we have that $\Rout = 1 - \delout + \frac{1}{n}$. Set $\delout$ to be $\varepsilon/2$, so $\Rout = 1 - \varepsilon/2 + o(1)$.

        Note that as $\Rin = k'/n'$, and $k'=\log q = \log (n+1)$, we have that $n' = (\log(n+1))/\Rin$. 

        We set $L = (\varepsilon/8)nn'$
        and set $B$ to be a sufficiently large constant. 
        Then we plug in $n'$ in the definition of $d$ the length of the encoding of our Gray code (\Cref{def:calg}) to get,
        \[
        d = n \frac{\log (n+1)}{\Rin} + 
        B \cdot (n + 1) + 2 L = \Theta(n\log n) \;,
        \]
        where the last equality follows as $B =O(1)$ 
        and that $2L = (\varepsilon/4)n((\log (n + 1))/\Rin)$.
        We now turn to compute the final rate according to \eqref{eq:calg-rate}. We get
        \begin{align*}
            \calr_{\calg} &\geq \frac{(1 - \delout) \cdot \Rin}{1 + (1+ \frac{1}{n})\cdot \frac{1}{\sqrt{n'}} + \frac{\varepsilon}{4}} \\
            &\geq \frac{(1 - \frac{\varepsilon}{2})\Rin}{1 + \frac{\varepsilon}{2}} \\
            & \geq \Rin - \varepsilon
        \end{align*}
        where the first inequality follows as there exists a large enough integer $n$ for which $(1 + 1/n)/\sqrt{n'} < \varepsilon/4$ (recall that $n' = O(\log n)$). 

        Now, note that by our choice of $L$, the failure probability of our algorithm is
        \[
        \Pr\left[ |j-\jhat|  \geq t \right] \leq \exp(-\Omega (t)) + \exp(-\Omega (n))\;. 
        \]
        We now show how to get the final failure probability as a function of $d$. Recall that $d = \Theta(n\cdot \log n)$, which implies that $n = \Theta(d/\log d)$. Indeed, let $C_1, C_2$ be constants such that $d/C_1 \leq n\log n \leq d/C_2$ for large enough $n$. It holds that
        
        \[
        n \geq \frac{d}{C_1 \log n} \geq \frac{d}{C_1 \log d - \log (C_2 \log n)} \geq \frac{d}{C_1\log d} \;,
        \]
        and a similar computation also shows that $n = O(d/\log d)$.
        Thus, 
        \[
        \Pr[|j-\jhat| \geq t ] \leq \exp(-\Omega(t)) + \exp\left(-\Omega\left(\frac{d}{\log d}\right)\right) \;.
        \]
        We proceed to show that the conditions given in \Cref{prop:decoding-alg-succeeds} indeed hold. 
        
         First, we observe that our choice of $L$ indeed satisfies $L = \omega(\log(kk') \log\log(kk'))$, as $L = \Theta(\varepsilon n \log n)$, which is much larger.  Next, we show that we can choose constants $\gapp, \xi$ so that \eqref{eq:decode-buffers-condition} and \eqref{eq:decode-outer-condition} hold, namely that
        \[ 
        2\exp(-C_p B) < \gapp < 1/4 \qquad \text{and} \qquad 2(1 + \xi)\pCinFail + 2\gapp < \delout\;.
        \]
        First, we choose a positive constant $\gapp < \min\{1/4, \frac{\delout}{4}\},$ recalling that $\delout = 1-\calr$ for our Reed-Solomon code $\Cout$, and thus $\delout$ is also a constant.  Thus the second inequality in \eqref{eq:decode-buffers-condition} is satisfied.
        Next, we note that by assumption, $\pCinFail = o(1)$ as $n' \to \infty$, and thus for large enough values of $n'$, we have $\pCinFail < \delout/8$; 
        then we can choose any $\xi < 1$ and satisfy \eqref{eq:decode-outer-condition} given that $\gapp < \delout/4$.  
        Finally, we may choose $B$ to be a sufficiently large constant (larger than $\ln(2/\gapp)/C_p$) and the first inequality in \eqref{eq:decode-buffers-condition} will hold as well.

        Finally, we analyze the final running time of our scheme given our choice of the outer code.
        Since RS codes can be encoded in time $O(n\cdot \text{poly}(\log n))$, plugging this in \Cref{prop:running-time} and noting that $n' = \log(n+1)/\Rin \leq  \log(d)/\Rin$, we get that the encoding of Gray code can be done in time
        \begin{align*}
        \tilde{O}(d^3) + O(d \cdot \text{poly}(\log d)) + O(d \cdot T_{\enc{\Cin}}(\log( d)/\Rin)) = \tilde{O}(d^3)\;, 
        \end{align*}
        where the last equality follows since $\Cin$ can be encoded in polynomial time.
        Further, the decoding time of RS codes is also $O(n\cdot \text{poly}(\log n))$, and thus, by our choice of $B$ and the fact that $n' = O(\log n)$, we get that the decoding time is
        \[
        \tilde{O}(d^2) + O(d \cdot \text{poly}(\log d)) + O\left( d \cdot T_{\dec{\Cin}}(n')\right)  = \tilde{O}(d^2) \;.
        \]
        The last equality above follows because, without loss of generality, we may assume that $T_{\dec{\Cin}}(n')$ is at most $\mathrm{poly}(n')\cdot 2^{k'}$, the running time of the brute-force maximum-likelihood decoder.  As we have $k' = \log(n+1)$, this is $n \cdot \mathrm{polylog}(n) = d \cdot \mathrm{polylog}(d)$, and hence the final term above is $\tilde{O}(d^2)$. 
    \end{proof}

\section*{Acknowledgements} We thank the Simons Institute for Theoretical Computer Science for their hospitality and support.

\bibliographystyle{alpha}
\bibliography{refs.bib}

\end{document}